\documentclass{article}

\usepackage{amsfonts}
\usepackage{amsthm}
\usepackage{amssymb}
\usepackage{mathtools}
\usepackage{dsfont}
\usepackage{bbm}
\usepackage{paralist}
\usepackage{multirow}
\usepackage{xspace}
\usepackage{graphicx}
\usepackage{tikz}
\usetikzlibrary{decorations.pathreplacing, shapes.geometric, positioning, arrows}
\usepackage[numbers]{natbib}
\usepackage{hyperref}
\usepackage{authblk}
\usepackage[noend,ruled]{algorithm2e}
\usepackage{colortbl}
\usepackage[rightcaption]{sidecap}
\sidecaptionvpos{figure}{t}

\newcommand{\np}{{\mathsf{NP}}}
\newcommand{\fpt}{{\mathsf{FPT}}}

\newcommand{\xp}{{\mathsf{XP}}}
\newcommand{\wone}{{\mathsf{W[1]}}}
\newcommand{\wtwo}{{\mathsf{W[2]}}}
\newcommand{\w}{{\mathsf{W}}}
\newcommand{\p}{{\mathsf{P}}}
\newcommand{\mytabref}[1]{\autoref{#1}}
\newcommand{\N}{{\mathbb{N}}}
\newcommand{\Q}{{\mathbb{Q}}}

\newtheorem{theorem}{Theorem}
\newtheorem{lemma}{Lemma}
\newtheorem{proposition}{Proposition}
\newtheorem{corollary}{Corollary}
\newtheorem{conjecture}{Conjecture}

\DeclareMathOperator{\cl}{cl}
\DeclareMathOperator{\dcl}{cl_{det}}
\DeclareMathOperator{\idcl}{\ensuremath{cl_{det}^{-1}}}
\DeclareMathOperator{\argmin}{argmin}

\makeatletter
\def\NAT@spacechar{~}% NEW
\makeatother

\newcommand{\probDef}[3]{
  \begin{quote}
    #1 \\
    \textbf{Input:} #2 \\
    \textbf{Question:} #3
  \end{quote}
}

\newcommand{\probSharpDef}[3]{
  \begin{quote}
    #1 \\
    \textbf{Input:} #2 \\
    \textbf{Compute:} #3
  \end{quote}
}

\newcommand{\probSetCover}{\textsc{Set Cover}\xspace}
\newcommand{\probDominatingSet}{\textsc{Dominating Set}\xspace}
\newcommand{\probColorClique}{\textsc{Multi-Colored Clique}\xspace}

\newcommand{\probSTConnectness}{\textsc{s-t~Connectedness}\xspace}
\newcommand{\probMaxIndependentSet}{\textsc{Independent Set}\xspace}
\newcommand{\probClosure}{\textsc{Maximum Weight Closure}\xspace}
\newcommand{\probEffectors}{\textsc{Effectors}\xspace}
\newcommand{\probCost}{\textsc{Effectors-Cost}\xspace}
 
\begin{document}

\title{The Complexity of Finding Effectors\thanks{An extended
    abstract appeared in \emph{Proceedings of the 12th Annual
      Conference on Theory and Applications of Models of Computation
      (TAMC '15)}, Volume~9076 of LNCS, pages 224--235, Springer, 2015. This article provides all proofs in full detail.}}

\iffalse
\author{Laurent Bulteau\thanks{Laurent.Bulteau@u-pem.fr,%
    IGM-LabInfo, CNRS UMR 8049, Universit\'e Paris-Est Marne-la-Vall\'ee, France.%
    Supported by the Alexander von Humboldt Foundation, Bonn, Germany.%
    Main work done while affiliated with  TU~Berlin.} \and Stefan Fafianie\thanks{stefan.fafianie@tu-berlin.de,%
    Institut f\"ur Informatik, Universit\"at Bonn, Germany. %
    Supported by the DFG Emmy Noether-program (KR 4286/1).%
    Main work done while affiliated with TU~Berlin.}}
  \fi
  
\author{Laurent Bulteau\thanks{(Laurent.Bulteau@u-pem.fr) Supported by the Alexander von Humboldt Foundation, Bonn, Germany.
  Main work done while affiliated with  TU~Berlin.}}
  
\affil{IGM-LabInfo, CNRS UMR 8049, Universit\'e Paris-Est Marne-la-Vall\'ee, France.}
    
\author{Stefan Fafianie\thanks{(stefan.fafianie@tu-berlin.de) Supported by the DFG Emmy Noether-program (KR 4286/1).
    Main work done while affiliated with TU~Berlin.}}

\affil{Institut f\"ur Informatik, Universit\"at Bonn, Germany.}

\author{Vincent Froese\thanks{(vincent.froese@tu-berlin.de) Supported by the DFG, project DAMM (NI 369/13).}}
\author{Rolf Niedermeier\thanks{(rolf.niedermeier@tu-berlin.de)}}
\author{Nimrod Talmon\thanks{(nimrodtalmon77@gmail.com) Supported by DFG Research Training Group ``Methods for Discrete Structures''~(GRK~1408).}}

\affil{Institut f\"ur Softwaretechnik und Theoretische Informatik, TU Berlin, Germany.}

\iffalse
\and Stefan Fafianie\thanks{Stefan
    Fafianie was supported by the DFG Emmy Noether-program (KR
    4286/1). Main work done while affiliated with TU~Berlin.}
    \and Vincent Froese\thanks{Vincent Froese was supported by the DFG, project DAMM (NI 369/13).}
    \and Rolf Niedermeier \and Nimrod Talmon\thanks{Nimrod Talmon was supported by DFG Research Training Group ``Methods for Discrete Structures''~(GRK~1408).}}

  \address{Laurent Bulteau \at
             IGM-LabInfo, CNRS UMR 8049, Universit\'e Paris-Est Marne-la-Vall\'ee, France.\\
             \email{Laurent.Bulteau@u-pem.fr}
               \and
             Stefan Fafianie \at
             Institut f\"ur Informatik,
             Universit\"at Bonn, Germany.\\
             \email{stefan.fafianie@tu-berlin.de}
               \and
             Vincent Froese, Rolf Niedermeier, and Nimrod Talmon \at
             Institut f\"ur Softwaretechnik und Theoretische Informatik,
             TU Berlin, Germany.\\
             \email{\{vincent.froese,rolf.niedermeier\}@tu-berlin.de, nimrodtalmon77@gmail.com}
           }
  \fi

\date{}

\maketitle

\begin{abstract}
The NP-hard \probEffectors problem on directed graphs is motivated by applications in network mining,
particularly concerning the analysis of probabilistic information-propagation processes in social networks.
In the corresponding model the arcs carry probabilities
and there is a probabilistic diffusion process activating 
nodes by neighboring activated nodes with probabilities as specified by the arcs. 
The point is to explain a given network activation state as well as possible by
using a minimum number of ``effector nodes'';
these are selected before the activation process starts.

We correct, complement, and extend 
previous work from the data mining community by a 
more thorough computational complexity analysis of \probEffectors,
identifying both tractable and intractable cases.
To this end, we also exploit a parameterization measuring the ``degree
of randomness'' (the number of `really' probabilistic arcs) 
which might prove useful for analyzing
other probabilistic network diffusion problems as well.
\end{abstract}

\section{Introduction}

To understand and master the dynamics
of information propagation in networks (biological, chemical,
computer, information, social) is a core research topic in data mining and related fields.
A prominent problem in this context is the $\np$-hard 
problem \probEffectors~\cite{LTGMH10}: The input is a 
directed (influence) graph with a subset of nodes 
marked as active (the target nodes) and each arc of 
the graph carries an 
influence probability greater than~0 and at most~1.
Assuming a certain diffusion process on the graph,
the task is to  find few ``effector nodes'' that can ``best explain'' 
the set of given active nodes,
that is, the activation state of the graph.

Specifically,
consider a set of nodes in the graph which are initially active.
Then,
due to a certain diffusion process,
several other nodes in the graph,
which initially were not active,
might become active as a result.
The diffusion model we consider
(and which is known as the independent cascade model~\cite{KKT15})
is such that,
at each time step,
a newly activated node
(initially only the chosen effectors are active)
has one chance to activate each non-active out-neighbor with the corresponding arc probability.
If an out-neighbor was successfully activated in the last time step,
then the propagation continues and this node has the chance to further activate its out-neighbors.
The propagation process terminates when there are no newly activated nodes.
\autoref{fig:intro-example} shows an example of a possible propagation process.
Given the activation state of the graph at the end of the propagation process,
we ask for the set of nodes, the effectors, which could best explain the current activation state.

\begin{figure}
  \centering
  \begin{tikzpicture}[>=stealth,scale=0.6]
      \tikzstyle{active}=[circle,draw,fill=black,minimum size=5pt,inner sep=3pt]
      \tikzstyle{inactive}=[circle,draw,minimum size=5pt,inner sep=3pt]

      \node at (0,1) {$t=0$};
      \node[active] (v1) at (0,6) {};
      \node[inactive] (v2) at (1.5,4) {};
      \node[inactive] (v3) at (-1.5,4) {};
      \node[inactive] (v4) at (0,2) {};

      \draw[->, very thick] (v1) -- (v2) node[midway, right] {0.5};
      \draw[->] (v1) -- (v3) node[midway, left] {0.8};
      \draw[->] (v2) -- (v3) node[midway, above] {0.1};
      \draw[->] (v3) -- (v4) node[midway, left] {1};
      \path[->] (v2) edge[bend right] node[midway, left] {0.3} (v4);
      \path[->] (v4) edge[bend right] node[midway, right] {0.9} (v2);
    \end{tikzpicture}
    \hspace{2em}
    \begin{tikzpicture}[>=stealth,scale=0.6]
      \tikzstyle{active}=[circle,draw,fill=black,minimum size=5pt,inner sep=3pt]
      \tikzstyle{inactive}=[circle,draw,minimum size=5pt,inner sep=3pt]

      \node at (0,1) {$t=1$};
      \node[active] (v1) at (0,6) {};
      \node[active] (v2) at (1.5,4) {};
      \node[inactive] (v3) at (-1.5,4) {};
      \node[inactive] (v4) at (0,2) {};

      \draw[->,help lines] (v1) -- (v2) node[midway, right] {0.5};
      \draw[->, help lines] (v1) -- (v3) node[midway, left] {0.8};
      \draw[->] (v2) -- (v3) node[midway, above] {0.1};
      \draw[->] (v3) -- (v4) node[midway, left] {1};
      \path[->, very thick] (v2) edge[bend right] node[midway, left] {0.3} (v4);
      \path[->] (v4) edge[bend right] node[midway, right] {0.9} (v2);
    \end{tikzpicture}
    \hspace{2em}
    \begin{tikzpicture}[>=stealth,scale=0.6]
      \tikzstyle{active}=[circle,draw,fill=black,minimum size=5pt,inner sep=3pt]
      \tikzstyle{inactive}=[circle,draw,minimum size=5pt,inner sep=3pt]

      \node at (0,1) {$t=2$};
      \node[active] (v1) at (0,6) {};
      \node[active] (v2) at (1.5,4) {};
      \node[inactive] (v3) at (-1.5,4) {};
      \node[active] (v4) at (0,2) {};

      \draw[->,help lines] (v1) -- (v2) node[midway, right] {0.5};
      \draw[->, help lines] (v1) -- (v3) node[midway, left] {0.8};
      \draw[->, help lines] (v2) -- (v3) node[midway, above] {0.1};
      \draw[->] (v3) -- (v4) node[midway, left] {1};
      \path[->, help lines] (v2) edge[bend right] node[midway, left] {0.3} (v4);
      \path[->, help lines] (v4) edge[bend right] node[midway, right] {0.9} (v2);
    \end{tikzpicture}
    \caption{An example depicting the information propagation according to the independent cascade model. The influence graph is a directed graph where the arcs are labeled with influence probabilities. Initially, at time~$t=0$, only the top node is active (black) and has a chance to independently activate the left and right node with the corresponding arc probabilities.
      In the example, the right node is activated (thick arc) while the left node is not. The probability of this event is thus~$0.5\cdot (1-0.8)$. The propagation then continues and the right node has a chance to activate its out-neighbors at time~$t=1$. Every activated node has only one chance (namely, after it became active the first time) to activate other inactive nodes. Note that at time~$t=2$ the bottom node cannot activate any new nodes. Hence, the propagation process terminates. The overall probability of this particular propagation (and of this particular activation state) equals $0.5\cdot(1-0.8)\cdot 0.3\cdot(1-0.1)=0.027$.
    }
  \label{fig:intro-example}
\end{figure}
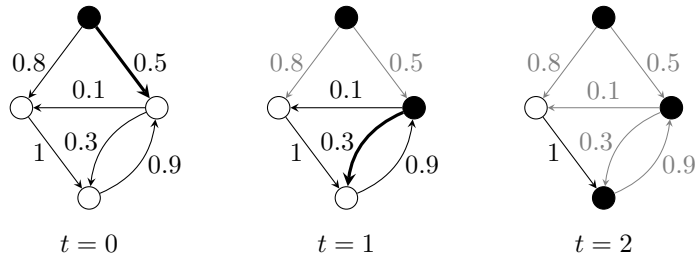

Being able to efficiently compute the set of effector nodes is helpful in many scenarios.
The paper by~\citet{LTGMH10} mentions several of them,
including being able to better understand how information propagates in social networks,
or finding those countries which are more prominent for spreading epidemics
(here, one might assume a graph where each country is a node,
and,
given the current state of some plague,
the effector nodes are those countries which explain this current state).
Motivated also by the scenario from epidemics,
one might be interested in providing shields against such plagues.
One possible way to achieve this is by finding the set of effectors,
and vaccinating the people in those countries.
Taking monetary costs into account,
it is desirable to find a small set of effectors;
thus,
in the \probEffectors problem,
the goal is to find a set of effectors of small size.
%Another example where the budget is useful to our problem,
%is when trying to efficiently secure a computer network against viruses.

It is important to note that we allow effectors to be chosen from
the \emph{whole} set of graph nodes and not only from the set of target nodes.
This makes our model, in a sense, more general than the original one by~\citet{LTGMH10}.\footnote{We
conjecture that both models coincide if we are allowed to choose an unlimited number of effectors,
that is,
if the number of chosen effectors does not matter.
On the contrary, they do not coincide if the number of effectors is bounded,
see~\autoref{sect:prelim}.}
See~\autoref{sect:prelim} for definitions of the main problems,
formal definition of our model, and a discussion about our model and its difference to that of~\citet{LTGMH10}.

Our main contribution is to extend and clarify research 
on the computational complexity status of \probEffectors, which has 
been initiated by~\citet{LTGMH10}.
In short,~\citet{LTGMH10} have shown that \probEffectors is generally $\np$-hard and hard to approximate,
developed an algorithm that is efficient on trees,
and used it to develop an efficient heuristic.
As \emph{probabilistic} information propagation is central in the independent cascade information-propagation model
which is in the heart of the \probEffectors problem (as well as in several other information-propagation models),
we put particular emphasis on studying how the ``degree of randomness''
in the network governs the computational complexity.
Moreover, compared to previous work, we make an effort to present
the results in a more formal setting, conducting a rigorous mathematical analysis.

Informally speaking (concrete statements of our results appear 
in \autoref{sect:prelim} after having provided formal definitions), 
we have gained the following main insights (also refer to \autoref{table:results} in
\autoref{sect:prelim}).
\begin{itemize}
  \item With unlimited degree of randomness, finding effectors is 
  computationally very hard. In fact, even computing the ``cost'' (how well 
  does a set of effectors explain a given activation state) of a \emph{given}
  set of effectors is intractable.
  This significantly differs from
  deterministic models.
  \item Even if the directed input graph is acyclic, then this does \emph{not}
  lead to a significant decrease of the computational complexity. 
  \item Bounding the degree of randomness (in other words, bounding the number 
  of arcs with probability different from~1), that is,
  \emph{parameterizing on the degree of randomness}, yields some 
  encouraging (fixed-parameter) tractability 
  results for otherwise intractable cases.
  \item We identify some flaws in the work of \citet{LTGMH10}
  (see~\autoref{sect:flawDetails} for details),
  who claim one case to be intractable which in fact is tractable and one case the other
  way around.
\end{itemize}

Admittedly, in real-world applications (where influence probabilities
are determined through observation and simulation, often involving noise) 
the number of probabilistic arcs may be high, thus, at first sight,
rendering the parameter ``number of probabilistic arcs'' doubtful.
However, note that finding effectors is computationally very hard
(also in terms of polynomial-time approximability; the approximation hardness of \probEffectors 
is mentioned by~\citet{LTGMH10} and follows, for example, from the reductions which use the \probSetCover problem).
So, in order to make the computation of a solution more feasible one might 
round up (to~1) arc probabilities which are close to~1 and round down (to~0) arc
probabilities which are close to~0. Thus, one can achieve a trade-off between running time and accuracy of the result.
Depending on the degree of rounding 
(as much as a subsequent fixed-parameter algorithm exploiting 
the mentioned parameter would ``allow''),
in this way one might at least find
a good approximation of an optimal set of effectors in reasonable time.

\paragraph{Related work.}
Our main point of reference is the work of~\citet{LTGMH10}.
Indeed, we use a slightly different problem definition:
They define the effectors to be necessarily a subset of the target nodes, 
whereas we allow the effectors to form an arbitrary
subset of the nodes. It turns out that these two definitions really
yield different problems, in the sense that a solution for one problem might not be a solution for the other
(see~\autoref{sec:Model} for an extensive discussion of the differences between these two models
and for an explanation on why we have chosen to define our model as it is defined).

The special case where all nodes are target nodes 
(and hence where the two models above clearly coincide) is
called \textsc{Influence Maximization} and is well studied in the
literature~\cite{BSKDSM07, DPRM01, KKT15}.
Specifically,
it is known that the \textsc{Influence Maximization} problem is $\np$-hard,
and a polynomial-time $(1 - 1 / e)$-approximation algorithm for this problem is given by~\citet{KKT15}.

Finally, a closely related deterministic version (called
\textsc{Target Set Selection}) with the 
additional difference of having node-individual 
thresholds specifying how many neighboring nodes need to be active 
to make a node active has also been extensively studied,
in particular from a parameterized complexity point of 
view~\cite{BCNS14-comp,BCNS14-jda,Ben-ZwiHLN11,CNNW14,NNUW13}.
\textsc{Target Set Selection} is $\np$-hard in general,
and hard to approximate, also in the parameterized sense
(specifically,
cannot be approximated even in $\fpt$-time (see~\autoref{sect:prelim}) with respect to the solution size).
It is $\np$-hard even on graphs of diameter $2$~\cite{NNUW13},
and it is tractable on some restricted graph classes such as trees~\cite{Ben-ZwiHLN11} and cliques~\cite{NNUW13}.

\section{Preliminaries}\label{sect:prelim}
In this section, we provide definitions used throughout the work.
We basically use the same definitions as~\citet{LTGMH10},
except for few differences in notation.

\paragraph{Graph Theory.}
We consider simple directed graphs~$G=(V,E)$ with a set~$V$ of nodes
and an arc set~$E\subseteq \{u\to v \mid u,v\in V, u\neq v\}$.
If there is an arc~$u\to v\in E$, then we call~$u$ an \emph{in-neighbor} of~$v$ and we call~$v$ an \emph{out-neighbor} of~$u$. For a subset~$V'\subseteq V$, we denote by~$G[V']:=(V', E')$ the subgraph of~$G$ induced by~$V'$, where $E':=\{u\to v\in E\mid u,v\in V'\}$.
An undirected graph~$G=(V,E)$ consists of a vertex set~$V$ and an edge set~$E\subseteq\{\{u,v\}\mid u,v\in V, u\neq v\}$.

We use the acronym DAG for directed acyclic graphs.
An undirected \emph{tree} is a connected acyclic graph.
A \emph{directed tree} is an arbitrary orientation of an undirected tree.
The \emph{condensation} of a directed graph~$G$ is a DAG containing a node~$v_C$ for each
strongly connected component~$C$ of~$G$ and there is an arc~$v_C\to v_{C'}$
if and only if there exists at least one arc from a node in~$C$ to a node in~$C'$.

\paragraph{Influence Graphs.}
An \emph{influence graph}~$G=(V,E,w)$ is a simple directed
graph equipped with a function~$w:E\rightarrow (0,1]\cap\Q$
assigning an \emph{influence weight} to each arc $u\to v\in E$
which represents the \emph{influence of node $u$ on node $v$}. Strictly speaking,
the influence is the probability that~$u$ propagates some information to~$v$.
We denote the number of nodes in $G$ by $n := |V|$ and the
number of arcs in $G$ by~$m := |E|$.

\paragraph{Information Propagation.}
We consider the following information-propaga\-tion process,
called the \emph{Independent Cascade (IC)} model \cite{KKT15}.
Within this model, each node is in one of two states: \emph{active} or \emph{inactive}.
When a node~$u$ becomes active for the first time, at time step~$t$,
it gets a single chance to activate its inactive out-neighbors.
Specifically, $u$ succeeds in activating a neighbor~$v$ with probability~$w(u \to v)$.
If $u$ succeeds, then $v$ will become active at step $t + 1$.
Otherwise, $u$ cannot make any more attempts to activate $v$ in any subsequent round.
The propagation process terminates when there are no newly activated nodes,
that is,
when the graph becomes static.

We remark that,
since our algorithms need to manipulate the probabilities determined by the function~$w$,
technically (and as usually)
we assume that the precision of the probabilities determined by this function 
is polynomially upper-bounded in the number~$n$ of nodes of the input graph,
and we ignore the time costs for adding or multiplying rational numbers assuming that these operations take constant time.
% Without this assumption,
% and if the precision of some of these probabilities were unbounded,
% the running times of our algorithms would be dependent on the exact way by which addition and multiplication
% of those numbers are being handled.

\paragraph{Cost Function.}
For a given influence graph~$G=(V,E,w)$,
a subset $X \subseteq V$ of effectors,
and a subset $A \subseteq V$ of active nodes,
we define a cost function
\[
  C_A(G,X) := \sum_{v \in A}\left(1-p(v|X)\right)+\smashoperator{\sum_{v \in V\setminus A}}p(v|X),
\]
where for each $v \in V$, we define $p(v|X)$ to be the probability of
$v$ being active after the termination of the information-propagation process starting
with~$X$ as the active nodes.
An alternative definition is that $C_A(G,X) := \sum_{v \in V} C_A(v,X)$,
where $C_A(v,X) := 1 - p(v|X)$ if $v \in A$ and $C_A(v,X) := p(v|X)$ if $v \notin A$.
One might think of this cost function as computing the expected number of nodes
which are incorrectly being activated or unactivated.

\paragraph{Main Problem Definition.}
Our central problem \probEffectors is formulated as a decision problem---it 
relates to finding few nodes which best explain (lowest cost) the given 
network activation state specified by a subset~$A\subseteq V$ of nodes.
\probDef
  {\probEffectors}
  {An influence graph $G = (V, E, w)$, a set of target nodes $A \subseteq V$, a budget $b \in \N$, and a cost $c \in \Q$.}
  {Is there a subset $X \subseteq V$ of effectors with $|X| \leq b$ and cost $C_A(G,X) \leq c$?}
We will additionally consider the related problem 
\probCost (see~\autoref{section:computingCost}) where the set~$X$ 
of effectors 
is already given and one has to determine its cost.

\paragraph{Parameters.}
The most natural parameters to consider for a parameterized computational complexity 
analysis are the maximum number~$b$ of
effectors, the cost value~$c$, and the number~$a:=|A|$ of target nodes.
Moreover, we will be especially interested in quantifying the amount of
randomness in the influence graph.
To this end, consider an arc $u \to v\in E$:
if $w(u \to v)= 1$,
then this arc is not probabilistic.
We define the parameter number~$r$ of probabilistic arcs, that is,
$r:= |\{u \to v\in E : w(u \to v) < 1\}|$.
% We will also briefly discuss the parameterization by the treewidth of the underlying 
% undirected graph.

\paragraph{Parameterized Complexity.}
We assume familiarity with the basic notions of algorithms and complexity.
Several of our results will be cast using the framework of parameterized complexity analysis.
An instance~$(I,k)$ of a parameterized problem consists of the classical
instance~$I$ and an integer~$k$ being the \emph{parameter} \cite{DF13,FG06,Nie06,Cyg15}.
A parameterized problem is called \emph{fixed-parameter tractable} (FPT) if there is an algorithm solving it in~$f(k)\cdot|I|^{O(1)}$ time, whereas an algorithm with running time~$|I|^{f(k)}$ only shows membership in the class~XP (clearly, FPT${}\subseteq{}$XP).
One can show that a parameterized problem~$L$ is (under certain complexity-theoretic assumptions) not fixed-parameter tractable by devising a \emph{parameterized reduction} from a
W[1]-hard or W[2]-hard problem (such as \textsc{Clique} or \textsc{Set
Cover}, respectively, each parameterized by the solution size) to~$L$.
A parameterized reduction from a parameterized problem~$L$ to another parameterized problem~$L'$ is a function that, given an instance~$(I,k)$, computes in~$f(k)\cdot |I|^{O(1)}$ time an instance~$(I',k')$~with~$k' \le g (k)$ such that~$(I,k)\in L \Leftrightarrow (I',k')\in L'$.
The common working hypothesis is that FPT${}\neq{}$W[1]. In fact, it is assumed that there is an infinite hierarchy
$$\text{FPT}\subset\text{W[1]}\subset\text{W[2]}\subset\ldots$$
called the $W$-hierarchy. Thus, for a parameterized problem to be~W[2]-hard is even stronger
in the sense that even if~FPT${}={}$W[1] holds, it is still possible that~FPT${}\neq{}$W[2].

\paragraph{Counting Complexity.}
We will also consider so called counting problems of the form ``Given~$x$, compute $f(x)$.'', where~$f$ is some function~$\{0,1\}^*\to\mathbb{N}$ (see \citet[Chapter~9]{AB09} for an introduction to counting complexity).
The class~$\#\p$ consists of all such functions~$f$ such that~$f(x)$ equals the number of accepting computation paths of a nondeterministic polynomial-time Turing machine on input~$x$.
Informally speaking, we can associate a decision problem in $\np$ (which asks weather there exists a solution or not) with a counting problem in~$\#\p$ (which asks for the number of solutions).
Clearly, if all counting problems in~$\#\p$ can be solved in polynomial time, then this implies $\p=\np$.
Analogously to $\np$-hardness, showing that a function
is $\#\p$-hard gives strong evidence for its computational intractability.
A function~$f:\{0,1\}^*\to\{0,1\}^*$ is $\#\p$-hard if a polynomial-time algorithm for~$f$ implies that all counting problems in~$\#\p$ are polynomial-time solvable.

\paragraph{Organization.}
Before we discuss our model and the one by \citet{LTGMH10}, 
we overview our main results in \autoref{table:results}.
We will treat the sub-problem \textsc{Eff\-ectors-Cost}
in~\autoref{section:computingCost},
and \probEffectors in~\autoref{section:findingEffectors}.
% Our results are summarized in~\autoref{table:results}.
Note that most of our results transfer to the model of \citet{LTGMH10}.
In particular, this implies that their claims that the ``zero-cost'' 
special case is $\np$-hard \cite[Lemma~1]{LTGMH10} and that the deterministic version is polynomial-time solvable
are both flawed, because from our results exactly the opposite follows
(see the last part of~\autoref{sect:flawDetails} for details).

  \newcommand{\with}[1]{ {\scriptsize $[$wrt.~$#1]$}}
  \newcommand{\see}[1]{ {\scriptsize \mytabref{#1}}}
  \newcommand{\seecite}[1]{ {\scriptsize \cite{#1}}}
  \newcommand{\clippedPath}[2][] {
    \begin{scope}
      \clip #2;
      \draw[#1] #2;
    \end{scope}
    \draw[white, line width=1pt] #2;
  }
  
  \newcommand{\resultBox}[2]{    
    \clippedPath[line width=4pt, #1!75!black, fill= #1!5]{#2}    
  }
  
  \newcommand{\titleBox}[1][1]{    
    \clippedPath[line width=3pt, black ]{(-1.4,0) -- (-1.4,#1) -- (0,#1) -- (0,0) --cycle }        
  } 
  \newcommand{\titleNode}[2][0.5]{    
    \node[rectangle,align=center] () at (-0.7,#1) {#2};    
  }  
  
  \newcommand{\headBox}[2]{    
    \clippedPath[line width=3pt, black ]{(#1,0) -- (#2,0) -- (#2,1.1) -- (#1,1.1) --cycle }        
  }   
  \newcommand{\headNode}[2]{    
    \node[multiline] () at (#1,0.55) {#2};    
  } 
 \begin{table}[t]
  \centering
  \caption{Computational complexity of the different variants of \probEffectors.
    Note that all hardness results hold also for DAGs.
    The parameter $a$ stands for the number of active nodes,
    $b$ for the budget,
    $c$ for the cost value,
    and $r$ for the number of probabilistic arcs.}
  \label{table:results}
  \begin{tikzpicture}[>=stealth, xscale=1.65,yscale=0.9]
    \tikzstyle{multiline}=[rectangle,align=center]
    
    \begin{scope}[yshift=1cm] 
      \headBox02
      \headNode1{Deterministic\\{\footnotesize ($r=0$)}}
      \headBox24
      \headNode3{Parameterized\\{\footnotesize (by $r$)}}
      \headBox46
      \headNode5{Probabilistic\\{\footnotesize (arbitrary $r$)}}      
    \end{scope}
    \begin{scope}
      \titleBox \titleNode{{\sc Effectors-}\\{\sc Cost} };
      \resultBox{green}{(0,0) -- (0,1) -- (4,1) -- (4,0) --cycle }     
      \resultBox{red}{(4,0) -- (4,1) -- (6,1) -- (6,0) --cycle }
      %\node () at (1,0.5) {$\p$};
      \node () at (2,0.5) {$\fpt$\with{r},\see{thm:costFPTr}};
      \node () at (5,0.5) {$\#\p$-hard,\see{thm:costNP}};      
    \end{scope}
    
    \begin{scope}[yshift=-2cm]    
      \titleBox[2]   
      \titleNode[1] {\probEffectors \\ (general case)};
      \resultBox{red}{(0,0) -- (0,1) -- (2,1) -- (2,2) -- (6,2) -- (6,0) --cycle }
     % \resultBox{red}{(0,0) -- (0,2)  -- (6,2) -- (6,0) --cycle }    
      \node[multiline] () at (4.0,0.9) { $\w[2]$-hard\with{b+c}, \see{thm:combinedHardness} \\ $\w[1]$-hard\with{a+b+c}, \see{thm:combinedHardness} };
      \resultBox{blue}{(0,1) -- (0,2) -- (2,2) -- (2,1) --cycle }     
      \node[multiline] () at (1,1.5) {$\xp$\with{\min(a,b,c)},\\[-0.15em]\see{prop:generalZEROrXP}};
    \end{scope}

    \begin{scope}[yshift=-3cm]
      \titleBox   
      \titleNode {Infinite budget\\$(b=\infty)$};
      
      \resultBox{green}{(0,0) -- (0,1) -- (4,1) -- (4,0) --cycle }     
      \resultBox{red}{(4,0) -- (4,1) -- (6,1) -- (6,0) --cycle }
%       \node () at (1,0.5) {$\p$};
      \node () at (2,0.5) {$\fpt$\with{r},\see{thm:inftyFPTr}};
      \node[multiline] () at (5,0.5) {\small $\np$-hard,\see{thm:inftyNP}\\ \small open: $\fpt$[wrt. $a$ or~$c$]}; 
    \end{scope}

    \begin{scope}[yshift=-5cm]    
      \titleBox[2]   
      \titleNode[1] {Influence\\Maximization\\$(A=V)$};
      \resultBox{red}{(0,0) -- (0,1) -- (2,1) -- (2,2) -- (6,2) -- (6,0) --cycle }    
      \node[multiline] () at (3.75,1) {$\w[1]$-hard\with{\min(b,c)},\see{thm:infmax}};
      \resultBox{green}{(0,1) -- (0,2) -- (2,2) -- (2,1) --cycle }     
      \node[multiline] () at (1,1.55) {$\fpt$\with{b+c},\\[-0.4em]\see{thm:infmax}};
      % \resultBox{green}{(0,0.5) -- (0,1.25) -- (2,1.25) -- (2,0.5) --cycle }     
      % \node[multiline] () at (1,0.88) {$\fpt$\with{\text{treewidth}},\\[-0.45em]\seecite{Ben-ZwiHLN11}};
    \end{scope}
  \end{tikzpicture}
\end{table}

\section{Model Discussion}
\label{sec:Model}
Our definition
of \probEffectors differs from the problem definition of
\citet{LTGMH10} in that we do not require the effectors to be chosen
among the target nodes.
Before pointing out possible advantages and motivating our problem definition,
we give a simple example illustrating the difference between these two definitions.

Consider the influence graph in \autoref{fig:ModelDiscExample},
consisting of one non-target node (white) having three outgoing arcs
with probability~1 each to three target nodes (black).
Clearly, for~$b=c=1$, this is a ``no''-instance if we are
only allowed to pick target nodes as effectors since
the probability of being active will be~0 for two of the three target
nodes in any case, which yields a cost of at least~2.
According to our problem definition, however, we are allowed to select the
non-target node, which only incurs a cost of~1, showing that this
is a ``yes''-instance.

Let us compare the two models.
First,
we think that our model captures the natural assumption that
an effector node does not have to remain active forever\footnote{Notably, in our model it actually remains active. The point is that before the whole computation starts (and after it ends) nodes may (have) become inactive again. Still, ``temporary activeness'' may make a node an effector that helps
explaining the currently observed network activation state.}.
Indeed, the modeling of \citet{LTGMH10} might be interpreted as a
``monotone version'' as for example discussed by \citet{ABS14}, 
while in this sense our model allows for ``non-monotone explanations''.
Second,
our model is more resilient to noise;
consider, for example, \autoref{fig:ModelDiscExample}.
It might be the case that indeed the top node is activated,
however,
due to noisy sampling methods,
it looks to us as if this top node is inactive.
In this simple example,
a solution according to the model of \citet{LTGMH10} would have to use three effectors to wrongly explain the data,
while a solution according to our model would be compute a correct and optimal solution with only one effector.

Clearly, if all nodes are target nodes (this particular setting is called \textsc{Influence Maximization}), 
then the two models coincide. 
Furthermore, we strongly conjecture that if we have an unlimited budget,
then it suffices to search for a solution
among the target nodes, that is, for~$b=\infty$, we believe that the two
problem definitions are also equivalent:

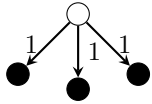
\begin{SCfigure}[2.3][t]
  \centering
  \begin{tikzpicture}[>=stealth, scale=0.8]
    \tikzstyle{active}=[circle,draw,fill=black,minimum size=5pt,inner sep=3pt]
    \tikzstyle{inactive}=[circle,draw,minimum size=5pt,inner sep=3pt]
    \node[inactive] (V1) at (1,1.5) {};
    \node[active] (V2) at (0,0.5) {};
    \node[active] (V3) at (1,0.25) {};
    \node[active] (V4) at (2,0.5) {};

    \draw[thick,->] (V1) -- (V2) node[midway, left] {1};
    \draw[thick,->] (V1) -- (V3) node[midway, right] {1};
    \draw[thick,->] (V1) -- (V4) node[midway, right] {1};
  \end{tikzpicture}
  \label{fig:ModelDiscExample}
  \caption{Example where it is optimal to choose a non-target node as effector.}
\end{SCfigure}

\begin{conjecture}
  \label{con:infinite-budget}
  For $b = \infty$, it holds that every ``yes''-instance~$(G,A,b,c)$ of
  \probEffectors has a solution~$X \subseteq A$.
\end{conjecture}

At least for directed trees
(that is, the underlying undirected graph is a tree---these also have
been studied by \citet{LTGMH10}) 
we can prove~\autoref{con:infinite-budget}.
The idea of proof is that if an optimal solution contains a non-target
node~$v$, then this node only influences nodes reachable from it via paths that do not visit other nodes in the solution. Within this smaller tree of influenced nodes there must be some subtrees rooted at target nodes such that the expected cost for such a subtree is smaller if its target root node is activated during the propagation process compared to the case when it is not. Choosing these target nodes directly as effectors, replacing the non-target node~$v$, yields another optimal solution with fewer non-target~nodes.

\begin{theorem}\label{thm:conjectureForTrees}
  \autoref{con:infinite-budget} holds for directed trees.
\end{theorem}

\begin{proof}
  Before proving the actual theorem, let us have a brief look on the
  probabilistics of the information-propagation process in directed
  trees.
  Clearly, in any influence graph, a node~$v$ can activate another
  node~$u$ only if there is a directed path from~$v$ to~$u$.
  Note that in a directed tree this path is unique if it exists.
  Moreover, the probability~$p(u|X)$ only depends on
  those nodes~$v\in X$ that are connected to~$u$ by a directed path that
  contains no other node from~$X$.
  To see that this is true, consider a node~$v\in X$ such that all directed paths from~$v$ to~$u$ contain another node from~$X$. Then, on each of these paths the corresponding node $x\in X$ has only one chance to activate~$u$ via propagation along the path. Since~$v$ cannot ``re-activate''~$x$ ($x$ is already active from the beginning), the activation probability of~$u$ does not depend on~$v$. 
  For a node~$v\in V$, let $\cl(v)\subseteq V$ denote the \emph{closure} of~$v$,
  that is, the set of all nodes~$u\in V$ for which there exists a directed
  path from~$v$ to~$u$ (including~$v$ itself, that is, $v\in\cl(v)$).

  Let~$(G=(V,E,w),A,b,c)$ with~$b=\infty$ be an input instance of
  \probEffectors, where~$G$ is an arbitrary directed tree.
  Let~$X\subseteq V$ be an optimal solution with~$X \not\subseteq A$,
  that is, there exists a node~$x\in X\setminus A$.
  We show that there is an optimal solution~$X^*$ containing
  fewer non-target nodes than~$X$.
  More formally, we show that there exists a solution~$X^*$
  with~$x\not\in X^*$ and~$X^*\setminus A\subsetneq X\setminus A$ such that
  $C_A(G,X^*) \leq C_A(G,X)$. Recursively applying this argument
  then proves the theorem.

  First, note that if~$C_A(G,X')\leq C_A(G,X)$ holds for~$X':=X\setminus\{x\}$,
  then we are done. Thus, we can assume $C_A(G,X') >
  C_A(G,X)$, or, equivalently:
  \begin{align}
    &C_A(G,X') - C_A(G,X) > 0\nonumber \\
    \Leftrightarrow &\sum_{v \in A}(1-p(v|X'))+\smashoperator{\sum_{v \in V\setminus
        A}}p(v|X') - \Big(\sum_{v \in
      A}(1-p(v|X))+\smashoperator{\sum_{v \in V\setminus
        A}}p(v|X)\Big) > 0\nonumber \\
    \Leftrightarrow &\smashoperator{\sum_{v \in A}}\big(p(v|X)-p(v|X')\big) - \smashoperator{\sum_{v \in V\setminus A}}\big(p(v|X) - p(v|X')\big) > 0.\label{eq:C'-C}
  \end{align}
  Now, consider a node~$v$ that is not in the closure
  of~$x$. Clearly, it holds that~$p(v|X)=p(v|X')$ since there is no directed path
  from~$x$ to~$v$, and thus~$x$ cannot change the 
  probability of~$v$ becoming active during the information-propagation process.
  Therefore, if we let~$A_x := \cl(x)\cap A$ and~$\overline{A}_x :=\cl(x)\setminus A_x$,
  then Inequality~\eqref{eq:C'-C} can be rewritten as
  \begin{align}
    \smashoperator{\sum_{v \in A_x}}\big(p(v|X)-p(v|X')\big)-\smashoperator{\sum_{v \in
    \overline{A}_x}}\big(p(v|X) - p(v|X')\big) > 0.\label{eq:C'-C2}
  \end{align}
  For a directed tree~$G$, the subgraph~$T_x:=G[\cl(x)]$ induced by the
  closure of~$x$ is a rooted directed tree with root~$x$,
  where all the arcs are directed from~$x$ to the leaves (that is, an
  \emph{out-tree}).
  Moreover, for a node~$v\in A_x$, there is exactly one directed path
  from~$x$ to~$v$ in~$T_x$. Let~$A_x'\subseteq A_x$ be the subset of target
  nodes~$v$ in the closure of~$x$ such that the directed path from~$x$ to~$v$
  contains no other target node from~$A_x$.
  Then, we can write the closure of~$x$ as the disjoint union
  $\cl(x) = \bigcup_{v\in A_x'}\cl(v) \cup Z$,
  where~$Z:=\cl(x)\setminus(\bigcup_{v\in A_x'}\cl(v))$.
  Note that~$Z\subseteq\overline{A}_x$.
  Therefore, we can write Inequality~\eqref{eq:C'-C2} as
  \begin{align*}
    \smashoperator{\sum_{v\in A_x'}}\Big(\smashoperator{\sum_{u \in
        A_v}}\big(p(u|X)-p(u|X')\big)-
    &\smashoperator{\sum_{u \in
    \overline{A}_v}}\big(p(u|X) - p(u|X')\big)\Big) -\\
    &\smashoperator{\sum_{v\in Z}}\big(p(v|X)-p(v|X')\big) > 0.
  \end{align*}
  Note that~$p(v|X)\ge p(v|X')$ holds for all~$v\in V$
  since~$X'\subseteq X$, which yields
  \[\sum_{v\in Z}\big(p(v|X)-p(v|X')\big) \ge 0.\]
  Therefore, the following holds
  \begin{align}
    \smashoperator{\sum_{v\in A_x'}}\Big(\smashoperator{\sum_{u \in A_v}}\big(p(u|X)-p(u|X')\big)-\smashoperator{\sum_{u \in \overline{A}_v}}\big(p(u|X) - p(u|X')\big)\Big) > 0.\label{eq:C'-C3}
  \end{align}
  Now, let~$p(\overline{v}|X)$ denote the probability that a
  node~$v$ is not activated given that the nodes in~$X$ are active
  and let~$p(u|\overline{v},X)$ be the probability
  of~$u$ being activated given that~$v$ is inactive
  and the nodes in~$X$ are active.

  Note that, for $v\in \cl(x)$ and~$u\in\cl(v)$,
  the probability of $u$ being active conditioned on~$v$ does
  not depend on~$x$ since~$v$ lies on the directed path from~$x$
  to~$u$, that is, $p(u|v,X) = p(u|v,X')$ and~$p(u|\overline{v},X) =
  p(u|\overline{v},X')$.
  Hence, we have
  \begin{align*}
    p(u|X)&=p(u|v,X)p(v|X)+p(u|\overline{v},X)p(\overline{v}|X)\\
          &=p(u|v,X')p(v|X)+p(u|\overline{v},X')(1-p(v|X))
  \end{align*}
  and
  \[
    p(u|X')=p(u|v,X')p(v|X')+p(u|\overline{v},X')(1-p(v|X')).
  \]
  This yields
  \begin{align*}
    p(u|X)-p(u|X')&=p(u|v,X')\big(p(v|X)-p(v|X')\big) +
                  p(u|\overline{v},X')\big(p(v|X')-p(v|X)\big)\\
                  &=\big(p(v|X)-p(v|X')\big)\big(p(u|v,X')-p(u|\overline{v},X')\big).
  \end{align*}
  Thus, for each~$v\in A_x'$, we have 
  \begin{align*}
    \smashoperator{\sum_{u \in A_v}}\big(p(u|X)-p(u|X')\big)-&\smashoperator{\sum_{u \in
    \overline{A}_v}}\big(p(u|X) - p(u|X')\big)=\\
    \big(p(v|X)-p(v|X')\big)\Big(&\smashoperator{\sum_{u \in
        A_v}}\big(p(u|v,X')-p(u|\overline{v},X')\big) -\\
    &\smashoperator{\sum_{u \in
        \overline{A}_v}}\big(p(u|v,X')-p(u|\overline{v},X')\big)\Big).\stepcounter{equation}\tag{\theequation}\label{eq:inner-prod}
  \end{align*}
  In the following, let
  \[\delta_v(X'):=\smashoperator{\sum_{u \in A_v}}
     \big(p(u|v,X')-p(u|\overline{v},X')\big) -
     \smashoperator{\sum_{u \in \overline{A}_v}}\big(p(u|v,X')-p(u|\overline{v},X')\big).\]
  
  Consider now Inequality~\eqref{eq:C'-C3} again.
  Since the outer summation in Inequality~\eqref{eq:C'-C3} over all nodes $v \in A_x'$ is positive,
  there must be some nodes $v \in A_x'$ for which the summand (that
  is, the right-hand side product of Equation~\eqref{eq:inner-prod})
  is positive.
  Note that~$p(v|X)-p(v|X')\ge 0$ since~$X'\subseteq X$ for all~$v\in
  A_x'$.
  Hence, the set~$A^*:=\{v\in A_x'\mid \delta_v(X') > 0\}$ is non-empty
  since these are the nodes for which the above product is positive.
  Furthermore, we define the new set of effectors~$X^*:=X\setminus\{x\}\cup A^* = X'\cup A^*$,
  which does not include the non-target node~$x$.

  Now, consider the difference~$C_A(G,X)-C_A(G,X^*)$.
  Since~$\{x\}\cup A^*\subseteq \cl(x)$, it follows~$p(v|X)=p(v|X^*)$
  for all~$v\not\in\cl(x)$.
  Thus, analogously to the above steps, we can write
  $C_A(G,X)-C_A(G,X^*)$ as
  \begin{align*}
    \smashoperator{\sum_{v\in A_x'}}\Big(\smashoperator{\sum_{u \in
        A_v}}\big(p(u|X^*)-p(u|X)\big)-
    &\smashoperator{\sum_{u \in
    \overline{A}_v}}\big(p(u|X^*) - p(u|X)\big)\Big) -\\
    &\smashoperator{\sum_{v\in Z}}\big(p(v|X^*)-p(v|X)\big).
  \end{align*}
  Note that, for each~$v\in Z$, it holds for
  all~$u\in A^*$ that~$v\not\in \cl(u)$.
  Hence, $p(v|X^*) = p(v|X')$, which implies
  \[\smashoperator{\sum_{v\in Z}}\big(p(v|X^*)-p(v|X)\big)\le 0.\]
  Thus, we obtain the following inequality
  \begin{align}
    C_A(G,X)-&C_A(G,X^*) \ge\nonumber\\
    &\smashoperator{\sum_{v\in A_x'}}\Big(\smashoperator{\sum_{u \in
        A_v}}\big(p(u|X^*)-p(u|X)\big)-
    \smashoperator{\sum_{u \in
    \overline{A}_v}}\big(p(u|X^*) - p(u|X)\big)\Big).\label{eq:C-C*1}
  \end{align}
  As in Equation~\eqref{eq:inner-prod}, we can rewrite the right-hand side of Inequality~\eqref{eq:C-C*1} to
  \begin{align}
    \smashoperator{\sum_{v\in A_x'}}\big(p(v|X^*)-p(v|X)\big)\Big(&\smashoperator{\sum_{u \in
        A_v}}\big(p(u|v,X)-p(u|\overline{v},X)\big) -\nonumber\\
      &\smashoperator{\sum_{u \in \overline{A}_v}}\big(p(u|v,X)-p(u|\overline{v},X)\big)\Big).\label{eq:C-C*2}
  \end{align}
  Clearly, for~$v\in A_x'$, the probability of~$u\in\cl(v)$ being
  active conditioned on~$v$ does not depend on~$x$,
  that is, it holds
    $p(u|v,X) = p(u|v,X')$ and
    $p(u|\overline{v},X) = p(u|\overline{v},X')$.
  By substituting these probabilities into~\eqref{eq:C-C*2} we
  arrive at

  \[
    C_A(G,X)-C_A(G,X^*) \ge \smashoperator{\sum_{v\in A_x'}}\big(p(v|X^*)-p(v|X)\big)\delta_v(X').
  \]

  Now, for each node~$v\in A^*$, it holds~$\delta_v(X')>0$ and
  $p(v|X^*)-p(v|X) = 1 - p(v|X) \ge 0$, and thus
  $\big(p(v|X^*)-p(v|X)\big)\delta_v(X')\geq 0.$
  For each~$v\in A_x'\setminus A^*$, it holds~$\delta_v(X')\le 0$ and
  $p(v|X^*)=p(v|X')\le p(v|X)$,
  and thus $\big(p(v|X^*)-p(v|X)\big)\delta_v(X')\geq 0.$
  
  Hence, $C_A(G,X)-C_A(G,X^*) \geq 0$ and, clearly, $X^*\setminus A
  \subsetneq X\setminus A$, and we are done.   
\end{proof}

The last theorem shows that our model for the \probEffectors problem and that of~\citet{LTGMH10} sometimes coincide.
In general,
however,
it is not completely clear how the computational complexity of our model for the \probEffectors problem
differs from that of~\citet{LTGMH10}.
We do \emph{mention} that our algorithmic results
(\autoref{lem:det_c0_p},
\autoref{prop:generalZEROrXP},
\autoref{thm:inftyFPTr})
easily transfer to the model of~\citet{LTGMH10},
as well as \autoref{thm:infmax}.

\section{Computing the Cost Function}\label{section:computingCost}
We consider the problem of computing the cost for a given set of effectors.
\probSharpDef
  {\probCost}
  {An influence graph $G = (V, E, w)$, a set of target nodes
    $A\subseteq V$, and a set of effectors $X \subseteq V$.}
  {The cost $C_A(G, X)$.}

\probCost is polynomial-time solvable on directed trees~\cite{LTGMH10}.
By contrast, 
\probCost is unlikely to be polynomial-time solvable even on DAGs.
This follows from a result by~\citet[Theorem 1]{wang2012scalable}.
They show that computing the expected number of activated nodes for a
single given effector is $\#\p$-hard on DAGs.
Note that for the case~$A=\emptyset$ (that is, $a=0$), the cost equals the expected
number of activated nodes at the end of the propagation process.
Hence, we obtain the following corollary of~\citet{wang2012scalable}.

\begin{corollary}\label{thm:costNP}
  \probCost on directed acyclic graphs is $\#\p$-hard even for~$a=0$ and~$|X|=1$.
\end{corollary}

Note that \autoref{thm:costNP} implies that \probCost on DAGs is not fixed-parameter tractable
with respect to the combined parameter~$(a,|X|)$.

On the positive side, \probCost is fixed-parameter tractable with respect to the
number~$r$ of probabilistic arcs.
The general idea is to recursively simulate the propagation process,
branching over the probabilistic arcs,
and to compute a weighted average of the final activation state of the graph.

\begin{theorem}\label{thm:costFPTr}
  Given an instance~$(G=(V,E),A,X)$ of \probCost, the probability~$p(v|X)$ for a given node~$v\in V$ can be computed in~$O(2^r \cdot n (n + m))$ time, where~$r$ is the number of probabilistic arcs.

  Accordingly, \probCost can be solved in $O(2^r \cdot n^2 (n + m))$ time.
\end{theorem}

\begin{proof}
  The overall idea of the proof is as follows.
  For each subset of the probabilistic arcs,
  we compute the cost,
  conditioned on the event that the propagation process was successful on these arcs,
  but not successful on the other probabilistic arcs.
  For each such subset we also compute the probability that this event happens.
  Then,
  by applying the law of total probability,
  it follows that the overall cost equals
  the weighted average of these conditioned costs,
  weighted by the probability of these events.
  
  We present the algorithm in a recursive way,
  mainly for the sake of having a formal proof for its correctness.
  To this end,
  let~$(G=(V,E),A,X)$ be an input instance of \probCost.
  Note that in order to compute the cost~$C_A(G,X)$, we
  compute the probability~$p(v|X)$ for each node~$v\in V$,
  because given all these probabilities
  it is straightforward to compute the cost in polynomial time.
  Hence, we prove the theorem by showing that computing~$p(v|X)$
  is fixed-parameter tractable with respect to~$r$
  using a search-tree algorithm that computes~$p(v|X)$
  for a given node~$v$ by recursively ``simulating'' all possible scenarios which
  could appear during the propagation process.
  
  To this end,
  we define an auxiliary function
  $\tilde{p}(v,X,F)$ denoting the probability that~$v$ is
  activated during the propagation process given that exactly the
  nodes in~$X$ are active but only the nodes in~$F\subseteq X$ are
  allowed to activate further nodes in the next step,
  whereas the nodes in~$X\setminus F$ can never activate any other node
  (indeed,~$p(v|X)=\tilde{p}(v,X,X)$).

  We now show how to compute~$\tilde{p}(v,X,F)$.
  First, if~$v\in X$, then $\tilde{p}(v,X,F)=1$, as it is already activated.
  Otherwise, if~$v\not\in X$ and~$F$ is closed (that is,~$F$ has no
  outgoing arcs to~$V\setminus X$),
  then there is no propagation at all and thus~$\tilde{p}(v,X,F)=0$.
  Otherwise, if~$X$ is not closed, then let~$N\subseteq V\setminus X$ denote
  the set of nodes in~$V\setminus X$ that have an
  incoming arc from some node in~$F$.
  Further, let~$N_d\subseteq N$ be the set of nodes that
  have at least one deterministic incoming arc from~$F$, and
  let~$N_p:=N\setminus N_d$.
  Also, let~$E_p\subseteq E$ be the set of probabilistic arcs from~$F$ to~$N$.
  Clearly, all nodes in~$N_d$ will be active in the next step of the
  propagation process, while the nodes in~$N_p$ will be active in the
  next step only with some positive probability.
  We can use the law of total probability on the subsets of $N_p$,
  and write
  \[\tilde{p}(v,X,F)=\smashoperator{\sum_{R\subseteq N_p}}\tilde{p}(v,X_R,F_R)q(X_R|X),\]
  where~$X_R:=X\cup F_R$ denotes the set of active nodes in the next time step,
  $F_R:=N_d\cup R$ denotes the set of newly active nodes in the next time step,
  and~$q(X_R|X)$ denotes the probability that \emph{exactly} the nodes
  in~$X_R$ are active in the next step given that \emph{exactly} the nodes
  in~$X$ are active.
  Note that, for each subset~$R\subseteq N_p$,
  \[q(X_R|X) = \prod_{u\in R}\Big(1-\overline{p}_u\Big)\smashoperator{\prod_{u\in N_p \setminus R}}\overline{p}_u,
  \text{  where  } \overline{p}_u:=\smashoperator{\prod_{v\to u\in E_p}}(1-w(v\to u)),\]
  is polynomial-time computable.
  As a result,
  we end up with the following recursive formula:
  \[\tilde{p}(v,X,F) :=
  \begin{cases}
    1,\; \text{ if }v\in X\\
    0,\; \text{ if }v\not\in X \text{ and } X \text{ closed}\\
    \sum_{R\subseteq N_p}q(X_R|X)\cdot \tilde{p}(v,X_R,F_R),\; \text{ else}.
  \end{cases}\]
  
  \begin{algorithm}[t]
    \normalsize
    \SetAlgoNoLine
    \caption{Pseudocode for $\tilde{p}(v, X, X)$.}
    \label{alg:costFPTr}
    \If{$v \in X$}{
      \Return $1$ \\
    }
    \If{$v \notin X$ and $X$ is closed}{
      \Return $0$ \\
    }
    \ForEach{$R\subseteq N_p$}{
      compute $q(X_R|X)$ \\
      compute $\tilde{p}(v,X_R,F_R)$ recursively \\
    }
    \Return $\sum_{R\subseteq N_p}q(X_R|X)\cdot \tilde{p}(v,X_R,F_R)$
  \end{algorithm}
  
  \autoref{alg:costFPTr} presents the pseudocode for computing $\tilde{p}$.
  For the running time, consider the recursion tree corresponding to
  the computation of~$\tilde{p}(v,X,X)$, where each vertex corresponds
  to a call of~$\tilde{p}$.
  
  For the running time,
  note that the inner computation (that is, without further recursive calls) of each node in the recursion tree
  can be done in time $O(n + m)$.  
  Moreover, for each call, either at least one node is inserted to~$X$,
  or the recursion stops.
  Therefore, the height of the recursion tree is upper-bounded by the number~$n$ of nodes.
  Lastly, each leaf in the recursion tree corresponds to a distinct subset of the probabilistic arcs,
  specifically, to those probabilistic arcs along which the propagation process carried on.
  Since there are~$2^r$ different subsets of probabilistic arcs,
  it follows that the number of leaves of the recursion tree is upper-bounded by~$2^r$.
  Thus, the overall size of the recursion tree is upper-bounded by $2^r \cdot n$,
  and hence,
  the running time is~$O(2^r \cdot n (n + m))$.
\end{proof}

\section{Finding Effectors}\label{section:findingEffectors}
We treat the general variant of \probEffectors in~\autoref{section:noAssumptions}, 
the special case of unlimited budget in~\autoref{section:infiniteBudget},
and the special case of influence maximization 
in~\autoref{section:everythingActive}.

\subsection{General Model}\label{section:noAssumptions}
We study how the parameters number~$a$ of target nodes, budget~$b$, and cost value~$c$ influence the computational complexity of \probEffectors.
We first observe that if at least one of them equals zero, then
\probEffectors is polynomial-time solvable.
This holds trivially for parameters $a$ and $b$; simply choose the
empty set as a solution.
This is optimal for $a = 0$, and the only feasible solution for $b = 0$.
For parameter $c$, the following holds, using a simple decomposition into 
strongly connected components.

\begin{lemma}\label{lem:det_c0_p}
  For $c = 0$, \probEffectors can be solved in linear time.
\end{lemma}

\begin{proof}
  If there is a directed path from a target node to a non-target node,
  then we have a ``no''-instance.
  Now every target node must be activated with probability~1,
  which is only possible along deterministic arcs.
  Let $G'$ be the condensation (that is, the DAG of strongly connected
  components) of the influence graph~$G$ after
  removing all probabilistic arcs.
  Then, we consider only the strongly connected components which
  contain at least one target node
  (note that all nodes in this component must be targets).
  Finally, if there are more than~$b$ of these target components that
  are sources in~$G'$, then we have a ``no''-instance.
  Otherwise, we arbitrarily pick a node from each component corresponding to a source,
  and return a positive answer.
  Each step requires linear time. 
\end{proof}

\noindent 
Based on \autoref{lem:det_c0_p},
by basically checking all possibilities in a brute-force manner,
we obtain simple polynomial-time algorithms for
\probEffectors in the cases of a constant number~$a$ of target
nodes, budget~$b$, or cost~$c$.

\begin{proposition}\label{prop:generalZEROrXP}
  For $r = 0$, \probEffectors is in $\xp$ with respect to each of
  the parameters~$a$, $b$, and~$c$.
\end{proposition}

\begin{proof}
  Containment in $\xp$ for the parameter $b$ is straightforward:
  For each possible set of effectors, we compute the cost in linear time and
  then return the best set of effectors.

  Note that for the case $r = 0$, we can assume that $b \leq a$.
  To see this, let~$X\subseteq V$ be a solution of size~$|X|> a$ and let~$A'\subseteq A$ be the subset of target nodes that are activated by choosing~$X$.
  Clearly, choosing~$A'$ as effectors is a better solution since it activates the
  same target nodes and only activates a subset of the non-target nodes activated by~$X$.
  Therefore, we also have containment in $\xp$ with respect to~$a$.

  It remains to show the claim for parameter~$c$.
  First, we choose which $c'\le c$ nodes incur a cost.
  Among these nodes, we set the target nodes to be non-targets, and vice versa.
  Then, we run the polynomial-time algorithm of \autoref{lem:det_c0_p} with cost~0.
  We exhaustively try all possible $\sum_{c'=0}^c\binom{n}{c'} \in O(n ^ c)$ choices to find a
  positive answer and return a negative answer otherwise. 
\end{proof}

In the following, we show that, even for $r = 0$ and the influence
graph being a DAG, \probEffectors is
$\wone$-hard with respect to the \emph{combined} parameter $(a, b, c)$,
and even $\wtwo$-hard with respect to the \emph{combined} parameter $(b,c)$.

\begin{theorem}\label{thm:combinedHardness}
  \mbox{}
  \begin{enumerate}
    \item \probEffectors,
      parameterized by the combined parameter $(a, \allowbreak b, c)$,
      is $\wone$-hard,
      even if $r = 0$
      and the influence graph is a DAG.
    \item \probEffectors,
      parameterized by the combined parameter $(b, c)$,
      is $\wtwo$-hard,
      even if $r = 0$
      and the influence graph is a DAG.
  \end{enumerate}
\end{theorem}

\begin{proof}
  We begin with the first statement,
  namely,
  that \probEffectors,
  parameterized by the combined parameter $(a, \allowbreak b, c)$,
  is $\wone$-hard,
  even if $r = 0$ and
  $G$ is a DAG.
  We describe a parameterized reduction from the following $\wone$-hard problem~\cite{FHRV09}.
  \probDef
    {\probColorClique}
    {A simple and undirected graph $G = (V, E)$ with $k$ colors on the vertices and $k \in \N$.}
    {Is there a $k$-vertex clique with exactly one occurrence of each
    color in the clique?}

  \begin{figure}[t]
  \centering
  \begin{tikzpicture}[>=stealth]
    \tikzstyle{active}=[circle,draw,fill=black,minimum size=5pt,inner sep=3pt]
    \tikzstyle{inactive}=[circle,draw,minimum size=5pt,inner sep=3pt]
    
    \node[inactive] (v1) at (1,4) {};
    \node[rectangle, draw] (l1) at (1, 4.5) {$u$};
    \draw (v1) -- (l1);
    \node[inactive] (v2) at (1,3) {};
    \node [rectangle, draw] (l2) at (1, 3.5) {$v$};
    \draw (v2) -- (l2);
    \node at (1,2) {$\vdots$};
    \node[inactive] (vn) at (1,1) {};
    
    \node[inactive] (e1) at (3,4.5) {};
    \node [rectangle, draw] (l3) at (3,5.1) {$\{u, v\}$};
    \draw (e1) -- (l3);
    \node[inactive] (e2) at (3,3.5) {};  
    \node[inactive] (e3) at (3,2.5) {};
    \node at (3,1.5) {$\vdots$};
    \node[inactive] (em) at (3,0.5) {};
    
    \node [rectangle, draw] (l4) at (5,2) {$(c(u), c(v))$};
    \node[active] (c11) at (5,4) {};
    \draw (l4) -- (5, 2.6);
    \node[active] (c12) at (6,4) {};
    \node[active] (c13) at (7,4) {};
    \node at (8,4) {$\ldots$};
    \node[active] (c1c) at (9,4) {};
    
    \node[active] (c21) at (5,3) {};
    \node[active] (c22) at (6,3) {};
    \node[active] (c23) at (7,3) {};
    \node at (8,3) {$\ldots$};
    \node[active] (c2c) at (9,3) {};
    
    \node[active] (cc1) at (5,1) {};
    \node[active] (cc2) at (6,1) {};
    \node[active] (cc3) at (7,1) {};
    \node at (8,1) {$\ldots$};
    \node[active] (ccc) at (9,1) {};
    
    \draw[rounded corners=8pt,line width=1.2pt] (4.5,4.4) rectangle (9.5,3.6);
    \draw[rounded corners=8pt,line width=1.2pt] (4.5,3.4) rectangle (9.5,2.6);
    \node at (7,2) {$\vdots$};
    \draw[rounded corners=8pt,line width=1.2pt] (4.5,1.4) rectangle (9.5,0.6);
    
    \draw[thick,->] (e1) -- (v1);
    \draw[thick,->] (e1) -- (v2);
    \draw[thick,->] (e1) -- (4.5,3);
    
    \draw[decorate,decoration={brace,mirror,raise=6pt,amplitude=10pt}, thick] (0.8,4.2)--(0.8, 0.8);
    \node at (-0.3, 2.67) {vertex};
    \node at (-0.3, 2.43) {nodes};
    \draw[decorate,decoration={brace,mirror,raise=6pt,amplitude=10pt}, thick] (2.8,4.7)--(2.8, 0.3);
    \node at (1.7, 2.67) {edge};
    \node at (1.7, 2.43) {nodes};
    \draw[decorate,decoration={brace,raise=6pt,amplitude=10pt}, thick] (9.5,4.2)--(9.5, 0.8);
    \node at (10.8, 2.5) {$\binom{k}{2}$ pairs};
    \draw[decorate,decoration={brace,raise=6pt,amplitude=10pt}, thick] (4.6,4.4)--(9.4, 4.4);
    \node at (7, 5.2) {$\binom{k}{2}+k+1$};
    \end{tikzpicture}
  \caption{Illustration of the influence graph used in the reduction from \probColorClique.
  In this example arcs are shown for one of the edge nodes.
  An arc from an edge node to a set of color pair nodes is used to represent the $\binom{k}{2}+k+1$ arcs to all nodes for this color pair.
  All arcs have an influence weight of 1.}
  \label{fig:MultiColoredCliqueReduction}
\end{figure}
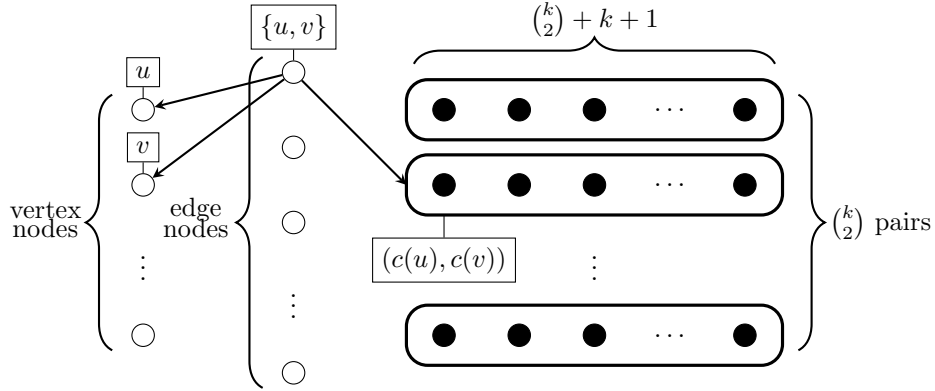

  Consider an instance $(G = (V, E), k)$ of \probColorClique.
  We assume that $k > 1$, otherwise the instance can be solved trivially.
  We construct an instance of \probEffectors with $b = \binom{k}{2}$, $c
  = \binom{k}{2} + k$ and an influence graph (see \autoref{fig:MultiColoredCliqueReduction} for an illustration) defined
  as follows.
  Add $\binom{k}{2} + k + 1$ nodes for each unordered pair of distinct colors.
  Let us call these nodes \emph{color-pair nodes}. These color-pair nodes are the target nodes~$A$, thus~$a = \binom{k}{2} \cdot (\binom{k}{2}+k+1)$.
  Now, add a \emph{vertex node} $n_v$ for each $v \in V$,
  add an \emph{edge node} $e_{u,v}$ for each $e = \{u, v\} \in E$,
  and add arcs $\{e_{u, v} \rightarrow n_u, e_{u, v} \rightarrow n_v\}$.
  For each edge~$e=\{u,v\}\in E$, let $L_e$ be the color-pair nodes corresponding to the
  colors of $u$ and $v$ and add arcs $\{e_{u, v} \to \ell \mid \ell
  \in L_e \}$.
  Finally, set the influence weights of all arcs to 1.
  
  Let $G'$ be the influence graph obtained by the above construction
  and notice that $G'$ is a DAG.
  We show that there is a $k$-vertex multi-colored clique in~$G$ if
  and only if there is a size-$b$ set of effectors that incurs a
  cost of at most~$c$ in~$G'$.
   
  Suppose that there is a multi-colored clique with $k$ vertices in
  $G$. Let $X$ be the edge nodes corresponding to the edges of this
  clique.
  Clearly, $|X| = \binom{k}{2} = b$. These effectors activate all
  color-pair nodes, that is, the complete target set $A$ with
  probability 1.
  Furthermore, the non-active edge and vertex nodes corresponding to the
  clique are activated, and a total cost of $\binom{k}{2} + k = c$ is
  incurred.
   
  For the reverse direction, let $X$ be a size-$b$ set of effectors
  that incurs a cost of at most $c$ in $G'$. Directly picking a vertex
  node is not optimal, since they are non-target nodes without outgoing arcs.
  Hence, they can only increase the cost.
  Also, without loss of generality, we can assume that~$X$ does not contain a color-pair node $x$.
  To see this, assume the contrary and
  suppose that~$X$ contains at least one edge node which 
  influences $x$. Then $X \setminus \{x\}$ is a solution with equivalent cost
  and smaller budget. In the other case, suppose that no such edge node is in $X$.   
  Then, we pay for at least $k+1$ other nodes corresponding to the
  same color-pair as $x$ since we can only take $b = \binom{k}{2}$ out of
  $\binom{k}{2} + k + 1$ nodes. Directly picking an edge node instead of $x$ incurs a cost of at most 3.
  By assumption, $k > 1$, that is, any optimal solution
  can be replaced by one that chooses only edge nodes as effectors.
  Now, in order to avoid a cost higher than $\binom{k}{2}+k = c$, 
  every color-pair node must be directly activated by an edge node.
  Then $X$ must contain exactly $\binom{k}{2}$ edge nodes, one for each color pair. A cost of at most $\binom{k}{2}+k$ is only obtained if they activate at most $k$ vertex nodes, i.e., the edges corresponding to the chosen edge nodes must form
  a multi-colored clique with $k$ vertices.

  We continue with the second statement,
  namely, that \probEffectors, parameterized by the combined parameter $(b, c)$,
  is $\wtwo$-hard, even if $r = 0$ and~$G$ is a DAG.
  We provide a parameterized reduction from the $\wtwo$-complete
  \probDominatingSet problem~\cite{DF13}.
  \probDef
    {\probDominatingSet}
    {A simple and undirected graph $G = (V, E)$, $k \in \N$.}
    {Is there a vertex subset $D \subseteq V$ such that $|D| \leq k$ and for each $v \in V$ either $v \in D$ or $\exists v' \in D$ such that $\{v, v'\} \in E$?}
  Consider an instance $(G = (V, E), k)$ of \probDominatingSet. We construct an instance for \probEffectors with $b = c = k$, and obtain the influence graph (see \autoref{fig:DominatingSetReduction} for an illustration) as follows:
  Add a node $i_v$ and a set of nodes $\{ c_{v,1} , \ldots, c_{v, k+1} \}$ for each vertex $v \in V$. Let us call these the \emph{initiator} and \emph{copies} of $v$, respectively. We connect each initiator of $v$ to all of its copies by adding arcs $\{ i_v \rightarrow c_{v,1} ,\ldots i_v \rightarrow c_{v, k+1} \}$. In a similar fashion, for each edge $\{u, v\} \in E$, we connect the initiator of $u$ to all copies of $v$ and vice versa. Finally, let the set of target nodes $A$ contain all copies of vertices and set the influence weight of all arcs to 1.
  
  \begin{figure}[t]
  \centering
  \begin{tikzpicture}[>=stealth, scale=0.8]
    \tikzstyle{active}=[circle,draw,fill=black,minimum size=5pt,inner sep=3pt]
    \tikzstyle{inactive}=[circle,draw,minimum size=5pt,inner sep=3pt]
    
    \node[inactive] (i1) at (0,4) {};
    \node[inactive] (i2) at (0,3) {};
    \node at (0,2) {$\vdots$};
    \node[inactive] (in) at (0,1) {};
    
    \node[active] (c10) at (3,4) {};
    \node[active] (c11) at (4,4) {};
    \node at (5,4.0) {$\ldots$};
    \node[active] (c1k) at (6,4) {};
    
    \node[active] (c20) at (3,3) {};
    \node[active] (c21) at (4,3) {};
    \node at (5,3) {$\ldots$};
    \node[active] (c2k) at (6,3) {};
    
    \node[active] (cn0) at (3,1) {};
    \node[active] (cn1) at (4,1) {};
    \node at (5,1) {$\ldots$};
    \node[active] (cnk) at (6,1) {};
     
    \draw[rounded corners=8pt,line width=1.2pt] (2.5,4.4) rectangle (6.5,3.6);
    \draw[rounded corners=8pt,line width=1.2pt] (2.5,3.4) rectangle (6.5,2.6);
    \node at (4.5,2) {$\vdots$};
    \draw[rounded corners=8pt,line width=1.2pt] (2.5,1.4) rectangle (6.5,0.6);
   
    \draw[thick,->] (i1) -- (2.5,4);
    \draw[thick,->] (i2) -- (2.5,3);
    \draw[thick,->] (in) -- (2.5,1);
    
    \draw[thick,->] (i1) -- (2.5,3);
    \draw[thick,->] (i2) -- (2.5,4);
    
    \draw[decorate,decoration={brace,mirror,raise=6pt,amplitude=10pt}, thick] (-0.2,4.2)--(-0.2, 0.8);
    \node[left] at (-1, 2.5) {initiators};
    
    \draw[decorate,decoration={brace,raise=6pt,amplitude=10pt}, thick] (2.6,4.4) -- (6.4,4.4);
    \node at (4.5,5.3) {$k+1$ copies};
  \end{tikzpicture}
  \caption{Illustration of the influence graph in the reduction from
    \probDominatingSet. The vertices corresponding to the two
    initiators at the top are neighbors in the input graph.
    An arc from an initiator to a set of copies is used to represent
    $k+1$ arcs, one to each copy.
    All arcs have an influence weight of 1.}
  \label{fig:DominatingSetReduction}
\end{figure}
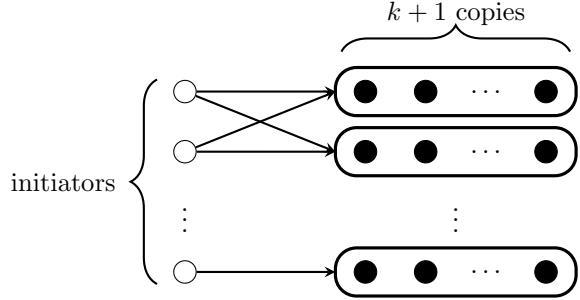
  
  Let $G'$ be the influence graph obtained in the construction and
  note that~$G'$ is a DAG.
  We show that there is a size-$k$ set~$D$ that dominates all vertices in~$G$ if and only if there is a size-$k$ set~$X$ of effectors that incurs a cost of at most~$k$ in~$G'$.
  Suppose that~$D$ is a $k$-dominating set for $G$. Let $X$ be the initiators of vertices in~$D$. These effectors activate all copies of vertices, i.e., the complete target set~$A$ with probability 1. Clearly, $|X| = k = b$ and a cost of $k = c$ is incurred for picking the initiators as effectors.
  
  For the reverse direction, let $X$ be a size-$k$ set of effectors
  that incur a cost of at most~$k$ in $G'$. Consider a solution in which we directly
  pick a copy~$x$ of a vertex~$v$ as an effector. Suppose that~$X$ contains the initiator of~$v$ or one of its neighbors. 
  Then $X \setminus \{x\}$ is a solution with equivalent cost and smaller budget. In the other case, suppose that $X$ contains no such initiator. Then, we pay for at least
  one other copy of $v$ since we can take at most $k$ out of $k+1$
  copies.
  Therefore, any optimal solution can be replaced by one that chooses
  only initiators as effectors. Now, every copy must be directly
  activated by an initiator to avoid a cost higher than $k$. 
  Furthermore, $X$ can contain at most $k$ initiators. These initiators can 
  only influence copies of their corresponding vertex or its neighbors,  
  that is, the vertices corresponding to the chosen initiators are a
  $k$-dominating set.   
\end{proof}

\subsection{Special Case: Unlimited Budget}\label{section:infiniteBudget}
Here, we concentrate on a model variant where we are allowed to choose any number of effectors,
that is, the goal is to minimize the overall cost with an unlimited budget of effectors.
In general, \probEffectors with unlimited budget remains intractable, though.

\begin{theorem}\label{thm:inftyNP}
  If $\p \neq \np$,
  then \probEffectors, even with unlimited budget, is not
  polynomial-time solvable on DAGs.
\end{theorem}

\begin{proof}
  We consider the following $\#\p$-hard~\cite{VLG1979A} counting problem.
  \probSharpDef
    {\probSTConnectness}
    {A directed acyclic graph $G = (V, E)$, two vertices $s, t \in V$.}
    {Number of subgraphs of $G$ in which there is a directed path from $s$ to $t$.}
  
  In the following, let~$\#_{st}(G)$ denote the number of subgraphs
  of~$G$ in which there exists a directed path from~$s$ to~$t$ (where distinct isomorphic subgraphs are considered different).
  We give a polynomial-time reduction from the decision version of
  \probSTConnectness, which asks whether~$\#_{st}(G)$ is at least
  a given integer~$z$.
  
  Let~$I=(G=(V,E), s, t, z)$ be an instance of the decision version of \probSTConnectness.
  We create an \probEffectors instance $I'=(G' = (V', E', w), \allowbreak A, b, c)$
  as follows.
  Let~$V_{st}\subseteq V$ be the set of vertices that
  lie on some directed path from~$s$ to~$t$ and let~$E_{st}\subseteq
  E$ be the set of arcs of all directed paths from~$s$ to~$t$.
  Further, let~$W:=V_{st}\setminus \{s,t\}$.
  Clearly, it holds~$\#_{st}(G) = \#_{st}(G[V_{st}]) \cdot~2^{|E\setminus E_{st}|}$ since~$\#_{st}(G[V\setminus V_{st}])=0$.
  Thus, in order to decide whether~$\#_{st}(G) \geq z$, we have to
  decide whether~$\#_{st}(G[V_{st}]) \geq z'$, where~$z':=\lceil z\cdot
  2^{-|E\setminus E_{st}|}\rceil$.
  
  We initialize $G'$ as the induced subgraph~$G[V_{st}]$
  and set~$w(v\to u):=1/2$ for each~$v\to u\in E_{st}$.
  We further create a copy $v'$ for each vertex $v\in W$,
  and add the arc~$v\to v'$ with $w(v\to v'):=1$.
  We also create a copy $s'$ of $s$,
  and add the arc~$s\to s'$ with $w(s\to s'):=1-p_{z'}$, where $p_{z'} :=
  z'\cdot 2^{-|E_{st}|}$.
  Finally, we set $A:=W\cup\{s\}$, $b = \infty$, and $c := |W|+1-2^{-|E_{st}|}$.
  The construction is illustrated in \autoref{fig:inftyNP}.
  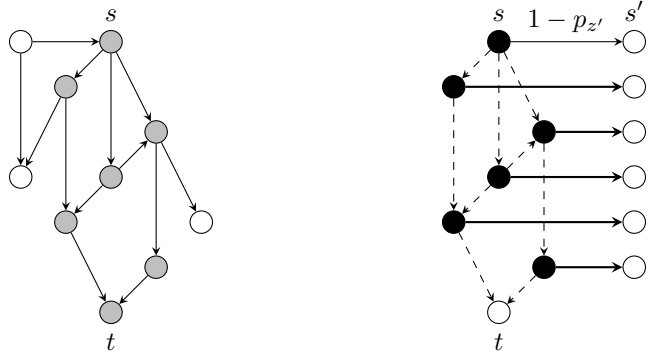
\begin{figure}[t]
    \centering
    \begin{tikzpicture}[>=stealth,scale=0.6]
      \tikzstyle{active}=[circle,draw,fill=black,minimum size=5pt,inner sep=3pt]
      \tikzstyle{inactive}=[circle,draw,minimum size=5pt,inner sep=3pt]

      % nodes of V_st
      \node[inactive,label=above:$s$,fill=lightgray] (s) at (0,6) {};
      \node[inactive,label=below:$t$,fill=lightgray] (t) at (0,0) {};
      \node[inactive,fill=lightgray] (v1) at (-1,5) {};
      \node[inactive,fill=lightgray] (v2) at (1,4) {};
      \node[inactive,fill=lightgray] (v3) at (0,3) {};
      \node[inactive,fill=lightgray] (v4) at (-1,2) {};
      \node[inactive,fill=lightgray] (v5) at (1,1) {};
      % other nodes
      \node[inactive] (v6) at (-2,6) {};
      \node[inactive] (v7) at (2,2) {};
      \node[inactive] (v8) at (-2,3) {};

      % arcs of E_st
      \draw[->] (s) -- (v1);
      \draw[->] (s) -- (v2);
      \draw[->] (s) -- (v3);
      \draw[->] (v1) -- (v4);
      \draw[->] (v2) -- (v5);
      \draw[->] (v3) -- (v2);
      \draw[->] (v3) -- (v4);
      \draw[->] (v4) -- (t);
      \draw[->] (v5) -- (t);
      % other arcs
      \draw[->] (v6) -- (s);
      \draw[->] (v6) -- (v8);
      \draw[->] (v1) -- (v8);
      \draw[->] (v2) -- (v7);        
    \end{tikzpicture}
    \hspace{8em}
    \begin{tikzpicture}[>=stealth,scale=0.6]
      \tikzstyle{active}=[circle,draw,fill=black,minimum size=5pt,inner sep=3pt]
      \tikzstyle{inactive}=[circle,draw,minimum size=5pt,inner sep=3pt]

      % nodes of V_st plus copies
      \node[active,label=above:$s$] (s) at (0,6) {};
      \node[inactive,label=above:$s'$] (s') at (3,6) {};
      \node[inactive,label=below:$t$] (t) at (0,0) {};
      \node[active] (v1) at (-1,5) {};
      \node[active] (v2) at (1,4) {};
      \node[active] (v3) at (0,3) {};
      \node[active] (v4) at (-1,2) {};
      \node[active] (v5) at (1,1) {};
      \node[inactive] (v1') at (3,5) {};
      \node[inactive] (v2') at (3,4) {};
      \node[inactive] (v3') at (3,3) {};
      \node[inactive] (v4') at (3,2) {};
      \node[inactive] (v5') at (3,1) {};

      % arcs of E_st
      \draw[->,dashed] (s) -- (v1);
      \draw[->,dashed] (s) -- (v2);
      \draw[->,dashed] (s) -- (v3);
      \draw[->,dashed] (v1) -- (v4);
      \draw[->,dashed] (v2) -- (v5);
      \draw[->,dashed] (v3) -- (v2);
      \draw[->,dashed] (v3) -- (v4);
      \draw[->,dashed] (v4) -- (t);
      \draw[->,dashed] (v5) -- (t);
      % arcs to copies
      % arcs of E_st
      \draw[->] (s) -- (s') node[midway, above] {$1-p_{z'}$};
      \draw[->,thick] (v1) -- (v1');
      \draw[->,thick] (v2) -- (v2');
      \draw[->,thick] (v3) -- (v3');
      \draw[->,thick] (v4) -- (v4');
      \draw[->,thick] (v5) -- (v5');
    \end{tikzpicture}
    \caption{Example illustrating the construction in the proof of~\autoref{thm:inftyNP}. 
      Left: A directed acyclic graph with two distinguished vertices~$s$ and~$t$, where the gray vertices lie on a directed~$s$-$t$-path.
      Right: The corresponding influence graph with target nodes colored in black.
      Dashed arcs have an influence weight of~1/2 and thick arcs have an influence weight of~1.
    }
    \label{fig:inftyNP}
  \end{figure}
  
  In the following, we prove two claims used to show the correctness
  of the above reduction.
  First, we claim that an optimal solution~$X$ of~$I'$ either
  equals~$\emptyset$ or~$\{s\}$.
  This can be seen as follows.
  Choosing~$s'$, $t$, or any copy $v'$ to be an effector is never optimal
  as these are all non-target nodes without outgoing arcs.
  Now, assume that~$X$ contains a node~$v\in W$ and let~$X':=X\setminus\{v\}$.
  Then, we have
  \begin{align*}
    C_A(G',X)-C_A(G',X') = &p(s|X')-p(s|X) + p(s'|X)-p(s'|X') +\\
    &\sum_{u\in W}(p(u|X') - p(u|X) + p(u'|X) - p(u'|X')) +\\
    &p(t|X)-p(t|X').
  \end{align*}
  Since~$G'$ is a DAG, it holds that there is no directed path from~$v$ to~$s$ and thus
  $p(s|X')=p(s|X)$ and consequently also~$p(s'|X')=p(s'|X)$,
  Moreover, note that $p(u|X)=p(u'|X)$ and~$p(u|X')=p(u'|X')$ holds for all~$u\in W$,
  and~$p(t|X) \ge p(t|X')$ clearly holds since~$X'\subseteq X$.
  Hence, $C_A(G',X)-C_A(G',X')\ge 0$ and
  therefore~$X'$ is also an optimal solution not containing~$v$,
  which proves the claim.

  Next, we claim that~$p(t|\{s\})=\#_{st}(G'[V_{st}])\cdot
  2^{-|E_{st}|}$.
  To prove this, we define an~\emph{$s$-$t$-scenario}~$S\subseteq
  E_{st}$ to be a subset of arcs such that~$\{s,t\}\subseteq V(S)$ and there
  is a directed path from~$s$ to each~$v\in V(S)$ in~$G[S]$.
  Let~$S^*:=\{v\to u\in E_{st}\mid v\in V(S)\}$
  denote the set of all outgoing arcs from nodes in~$V(S)$.
  We denote the set of all $s$-$t$-scenarios by~$\mathcal{S}_{st}$.
  Note that each scenario~$S$ constitutes a possible propagation
  in which exactly the arcs in~$S$ activated their endpoints and
  the arcs in~$S^*\setminus S$ did not activate their endpoints.
  The probability~$q(S)$ for a given $s$-$t$-scenario~$S$ to occur is
  thus~$2^{-|S^*|}$.
  Clearly, we can write
  \[p(t|\{s\})=\sum_{S\in\mathcal{S}_{st}}q(S)=\sum_{S\in\mathcal{S}_{st}}2^{-|S^*|}=
    2^{-|E_{st}|}\cdot\sum_{S\in\mathcal{S}_{st}}2^{|E_{st}\setminus S^*|}.\]

  Now, for a subset~$F\subseteq E_{st}$ of arcs where~$s$ is
  connected to~$t$ in the subgraph~$G'[F]$, let~$sc(F)$ denote
  the scenario~$S\in\mathcal{S}_{st}$ where~$S\subseteq F$ and
  $S\subsetneq S'$ for all~$S'\neq S\in \mathcal{S}_{st}$ such that~$S' \subseteq F$.
  It holds that~$F = S \cup F^*$, where~$S:=sc(F)$ and
  $F^*:=F\setminus S\subseteq E_{st}\setminus S^*$.
  Hence, we have~$\#_{st}(G'[V_{st}])=\sum_{S\in\mathcal{S}_{st}}2^{|E_{st}\setminus
    S^*|}$, which proves the claim.

  We now decide the instance~$I$ as follows.
  Note that~$C_A(G',\emptyset)=|W|+1$
  and~$C_A(G',\{s\})=|W|+1-p_{z'}+p(t|\{s\})$.
  Therefore, if~$I'$ is a ``yes''-instance, then~$\{s\}$ is the
  optimal solution with $|W|+1-p_{z'}+p(t|s) \leq c =
  |W|+1-2^{-|E_{st}|}$, which implies~$p_{z'} - p(t|\{s\}) \geq 2^{-|E_{st}|}$.
  It follows that $\#_{st}(G'[V_{st}]) < z'$.
  Therefore, $I$~is a ``no''-instance.
  If~$I'$ is a ``no''-instance, then~$p_{z'} - p(t|\{s\}) <
  2^{-|E_{st}|}$, which implies~$\#_{st}(G'[V_{st}]) \geq z'$,
  hence~$I$ is a ``yes''-instance.     
\end{proof}

With unlimited budget, however, \probEffectors becomes fixed-parameter
tract\-able with respect to the parameter number $r$ of probabilistic arcs.

\begin{theorem}\label{thm:inftyFPTr}
  If $b = \infty$, 
  then \probEffectors is solvable in $O(4^r\cdot n^4)$ time,
  where~$r$ is the number of probabilistic arcs.
\end{theorem}

\newcommand{\Xplus}{X_{o}}
\newcommand{\Xp}{X_p}
\newcommand{\Yplus}{Y_{o}}
\newcommand{\Yp}{Y_p}
\newcommand{\Vplus}{V_{o}}
\newcommand{\Vp}{V_p}

\begin{proof}
\begin{figure}[t]
  \centering
  \begin{tikzpicture}[>=stealth]
    \tikzstyle{active}=[circle,draw,fill=black,minimum size=5pt,inner sep=3pt]
    \tikzstyle{inactive}=[circle,draw,minimum size=5pt,inner sep=3pt]
    \tikzstyle{xaura}=[circle,draw,minimum size=5pt,inner sep=3.5pt, line width=3pt,color=black!40]
    \tikzstyle{prob}=[circle,draw,minimum size=5pt,inner sep=0pt]
                \foreach \x/\y [count = \xi] in {1.5/3,3.5/3}{
                   \node[xaura] () at (\x,\y) {};
                   \node[prob] (Vp\xi) at (\x,\y) {\sf X};
    }           
                \foreach \x/\y [count = \xi from 3] in {7.5/3,9.5/3}{
                   \node[prob] (Vp\xi) at (\x,\y) {\sf X};
    }   
                \foreach \x/\y [count = \xi] in {0/2, 1/5, 4.8/0.5, 5.5/1.5, 6.2/0.5,5.5/4, 8/5,8/1,10/5,10/1 }{
                   \node[inactive] (Vy\xi) at (\x,\y) {};
    }
                \foreach \x/\y [count = \xi] in {1/1,3/1,3/5,0/4,0/2, 4.5/5,6/5}{
                   \node[xaura] () at (\x,\y) {};
                   \node[inactive] (Vx\xi) at (\x,\y) {};
    }   
                \foreach \s/\t in {p1/x1,p1/x2,p2/x2,p4/y8,x4/x5, x3/p1,x3/p2, y2/x3, y3/y5,y5/y4,y4/y3,y6/p2, y7/p3,y7/p4,y9/p4,x6/x7} {
                                \draw[->] (V\s) -- (V\t);
                }
                \foreach \s/\t in {p1/y1,p2/y3,p2/p3,p3/y5,p3/y6, p4/y10} {
                                \draw[->, dashed] (V\s) -- (V\t);
                }
                
                \draw[rounded corners=8pt,line width=1.2pt] (0.4,3.6) rectangle (10.6,2.4);
                \node at (5.5,3.2) {$\Vp$};
        %       \draw[rounded corners=13pt] (0.5,3.5) rectangle (4.0,2.5);
                \draw[rounded corners=8pt,line width=1.2pt] (0.5,3.5) rectangle (4.0,0.5);
                \node at (2.5,3) {$\Xp$};
                \node at (1.7,1.8) {$\Xplus$};
                \draw[rounded corners=8pt,line width=1.2pt] (7.0,5.5) rectangle (10.5,2.5);
                \node at (8.5,3) {$\Yp$};
                \node at (9.2,4.2) {$\Yplus$};

                \node at (6.5,4) {$V'$};
                
  \end{tikzpicture}
  \caption{Illustration for~\autoref{thm:inftyFPTr}.
  Effectors of a solution are marked with an aura. 
  Probabilistic arcs are dashed, and nodes of $\Vp$ (with an outgoing probabilistic arc) are marked with a cross. 
  For readability, target nodes are not represented. Intuitively, the algorithm guesses the partition of $\Vp$ 
  into $\Xp$ (effectors) and $\Yp$ (non-effectors). Node set $\Xp$ (respectively,~$\Yp$) is then extended to its closure 
  $\Xplus$ (respectively, its closure $\Yplus$ in the reverse graph). The remaining nodes form a deterministic subgraph $G[V']$, 
  in which effectors, forming the set~$X'$, 
  are selected by solving an instance of {\sc Maximum Weight Closure}. }
  \label{fig:inftyFPTr}
\end{figure}
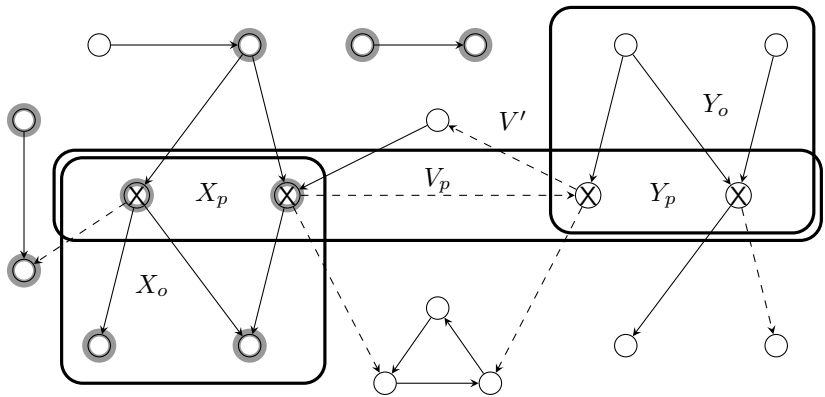

  The general idea is to fully determine the probabilistic aspects of the graph,
  and then to remove all of the corresponding nodes and arcs.
  We can show that this leaves an equivalent ``deterministic graph'' that we can solve using a reduction to 
  the problem \probClosure,
  which is itself polynomial-time solvable by a polynomial-time reduction to a flow maximization problem~\cite[Chapter~19]{AMO93}.
  \probSharpDef
    {\probClosure}
    {A directed graph $G = (V, E)$ with weights on the vertices.}
    {A maximum-weight set of vertices $X \subseteq V$  with no arcs going out of the set.}    
  We start with some notation (see~\autoref{fig:inftyFPTr} for an illustration).
  For an input graph $G = (V, E)$,
  let $E_p:=\{u\to v\in E\mid w(u\to v)<1\}$ denote the set of
  probabilistic arcs
  and let $\Vp := \{u \mid u \to v \in E_p\}$ denote the set of
  nodes with at least one outgoing probabilistic arc.
  For a node $v\in V$, let
  $\dcl(v)$
  ($\idcl(v)$)
  denote the set of all nodes $u$ such that there exists at least
  one deterministic path
  from $v$ to $u$
  (respectively, from $u$ to $v$),
  where a deterministic path is a path containing only deterministic arcs.
  We extend the notation to subsets $V'$ of $V$ and write
  $\dcl(V') = \bigcup_{v \in V'} \dcl(v)$
  and
  $\idcl(V') = \bigcup_{v \in V'} \idcl(v)$.
  We call a subset~$V'\subseteq V$ of nodes
  \emph{deterministically closed} if and only if~$\dcl(V')=V'$,
  that is, there are no outgoing deterministic arcs from~$V'$ to~$V\setminus V'$.
  
  Our algorithm will be based on a closer analysis of
  the structure of an optimal solution.
  To this end, let~$G=(V,E,w)$ be an input graph with a
  set~$A\subseteq V$ of target nodes and let~$X\subseteq V$ be an
  optimal solution with minimum cost~$C_A(G,X)$.
  Clearly, we can assume that~$X$ is deterministically closed, that
  is,~$\dcl(X)=X$, since we have an infinite budget~$b=\infty$.
         
        We write~$\Vp$ as a disjoint union of $\Xp:=\Vp\cap X$ and $\Yp:=\Vp\setminus X$.
	We also use $\Xplus:=\dcl(\Xp)$, $\Yplus:=\idcl(\Yp)$, and~$\Vplus=\Xplus\cup\Yplus$.
	Since $X$ is deterministically closed, we have that $\Xplus \subseteq X$ and $\Yplus\cap X = \emptyset$.
	We write $V':=V\setminus\Vplus$ and $X':=X\setminus \Xplus=X\cap V'$. 
	Note that $X'$ is deterministically closed in~$G[V']$ and that~$G[V']$ contains only deterministic arcs.
	Moreover, note that the sets $\Xplus$, $\Yp$, $\Yplus$, $\Vplus$, and~$V'$, are directly deduced from the choice of $\Xp$,
	and that for a given~$\Xp$, the set $X'$ can be any deterministically
        closed subset of $V'$.
  
  We first show that the nodes in $\Vplus$ are only influenced by effectors in~$\Xplus$, that is, for any node~$v\in \Vplus$, it holds that
  $p(v|X)=p(v|\Xplus)$. This is clear for $v\in\Xplus$, since in this case $p(v|X)=p(v|\Xplus)=1$.
  Assume now that there is a node~$x\in X'$ with a directed path to~$v\in \Yplus$ 
	that does not contain  any node from~$\Xplus$ (if every directed path from~$x$ to~$v$ passes through~$X_0$, then clearly~$x$ does not influence~$v$). 
	Two cases are possible, depending on whether this path is deterministic.
	If it is, then, since $v\in\idcl(\Yp)$, there exists a deterministic path from $x$ to some $u\in\Yp$, via $v$.
	Hence, $x\in \idcl(\Yp)=\Yplus$, yielding a contradiction.
	Assume now that the path from $x$ to~$v$ has a probabilistic arc and write $u\to u'$ for the first such arc. Hence, $x\in \idcl(u)$ and $u\in \Vp$. Since we assumed that the path does not contain any node from~$\Xplus$, we have $u\notin \Xp$, and therefore $u\in\Yp$. Again, we have $x\in \idcl(\Yp)$, yielding a contradiction. 
        Hence, the nodes in~$\Vplus$ are not influenced by the
        nodes in~$X'$.
        
   Now consider the nodes in $V'$. Note that we have~$p(v|X)=1$ for~$v\in X'$ and $p(v|X) = p(v|\Xplus)$
  for~$v\in V'\setminus X'$, since~$G[V']$ is deterministic and~$X'$~is
  deterministically closed.
   \iffalse
  Thus, we can write the cost~$C_A(G,X)$ as  
  \begin{align*}
    C_A(G,X) &= \smashoperator{\sum_{v\in \Vplus}}C_A(v,X) + \sum_{v\in
      V'}C_A(v,X)\\
    &= \smashoperator{\sum_{v\in \Vplus\cap A}}(1-p(v|\Xplus))+\smashoperator{\sum_{v\in \Vplus\setminus A}}p(v|\Xplus)
    + \smashoperator{\sum_{v\in V'\cap A}}(1-p(v|X)) + \smashoperator{\sum_{v\in V'\setminus A}}p(v|X).
  \end{align*}
\fi
    Overall, $C_A(v,X)=C_A(v,\Xplus )$ for all $v\in V\setminus X'$. The total cost of solution~$X$ can now be written as
  \begin{align*}
    C_A(G,X) &= \smashoperator{\sum_{v\in V\setminus X'}}C_A(v,\Xplus)
      + \sum_{v\in X'}C_A(v,X) \\
      &= \smashoperator{\sum_{v\in V}}C_A(v,\Xplus)
      - \sum_{v\in X'}(C_A(v,\Xplus)-C_A(v,X)) \\
       &= \alpha(\Xplus)-\beta(\Xplus,X'),
  \end{align*}
where
  \[
   \alpha(\Xplus):=\sum_{v\in V}C_A(v,\Xplus) \text{\quad and \quad} \beta(\Xplus,X'):= \sum_{v\in X'}(C_A(v,\Xplus)-C_A(v,X)).
  \]
  We further define, for all $v\in V'$,
  $\gamma(v,\Xplus):=1-p(v|\Xplus)$ if $v\in A$, and
  $\gamma(v,\Xplus)=p(v|\Xplus)-1$ if $v\notin A$.
  Note that, for $v\in X'$, the
  difference~$C_A(v,\Xplus)-C_A(v,X)$
  is exactly $\gamma(v,\Xplus)$, hence $\beta(\Xplus,X')= \sum_{v\in X'}\gamma(v,\Xplus)$.
 
  \iffalse
  Thus, by comparison with the solution where no vertex in $V'$ is activated, 
  activating a vertex~$v\in V'\cap A$ decreases the
  cost by~$1-p(v|\Xplus)$, while activating a vertex~$v\in V'\setminus A$
  increases the cost by~$1-p(v|\Xplus)$.
  
  An optimal solution~$X$ thus fulfills the following:
  \[X\cap \Vp = \argmin_{\Xp\subseteq \Vp}(\alpha(\Xp) - \beta(\Xp)),\]
  where we define
  \[\alpha(\Xp) := \smashoperator{\sum_{v\in \Vplus}}C_A(v,\Xplus),\]
  for
  \[\Xplus:=\dcl(\Xp),\quad  \Vplus:=\Xplus\cup \idcl(\Vp\setminus \Xp),\]
  and
  \[\beta(\Xp) := \max_{\substack{X'\subseteq V'\\\dcl(X')=X'}}\Big(\smashoperator{\sum_{v\in X'\cap A}}(1-p(v|\Xplus)) +
  \smashoperator{\sum_{v\in X'\setminus A}}(p(v|\Xplus)-1)\Big).\]  
  \fi

  \begin{algorithm}[t]
    \normalsize
    \SetAlgoNoLine
    \caption{\label{algo:inftyFPTr}Pseudocode for \probEffectors with $b=\infty.$}
    \label{alg:inftyFPTr}
    \ForEach{$\Xp\subseteq \Vp$}{
      compute $\Xplus:=\dcl(\Xp)$ \\
      compute $\Yplus:=\idcl(\Vp\setminus \Xp)$ \\
      \ForEach{$v \in V$}{
        compute $p(v|\Xplus)$ (using~\autoref{thm:costFPTr}) and $\gamma(v,\Xplus)$
      }
      compute $\alpha(\Xplus)$ \\
      compute $X'$ maximizing $\beta(\Xplus,X')$ using \probClosure on~$G[V']$, with weights $\gamma(v,\Xplus)$ \\
    }
    \Return the $\Xplus\cup X'$ which gives the minimum $\alpha(\Xplus) - \beta(\Xplus,X')$
  \end{algorithm}

  The algorithm can now be described
  directly based on the above formulas.
  Specifically, we branch over all subsets~$\Xp\subseteq \Vp$ (note that the number of these subsets is upper-bounded by $2^r$).
  For each such subset~$\Xp\subseteq \Vp$,
  we can compute $\Xplus$ and~$\Yplus$ in linear time because this
  involves propagation only through deterministic arcs
  (outgoing for $\Xplus$ and ingoing for $\Yplus$).
  Then, for each node $v\in V$, 
  we compute~$p(v | \Xplus)$ using~\autoref{thm:costFPTr}
  in $O(2^r \cdot n(n + m))$ time.
  This yields the values~$\alpha(\Xplus)$ and~$\gamma(v,\Xplus)$ for each $v\in V'$.
  By the discussion above, it remains to select a closed subset~$X'\subseteq V'$
  such that the cost $C_A(G,\Xplus\cup X')=\alpha(\Xplus) - \beta(\Xplus,X')$ is minimized. 
  That is, we have to select the subset~$X'$ that maximizes the value of~$\beta(\Xplus, X')$.
  Hence, the subset~$X'$ can be computed as the solution of \probClosure on $G[V']$ (which is solved by a maximum flow computation in $O(n^3)$ time),
  where the weight of any $v\in V'$ is~$\gamma(v,\Xplus)$.
  Finally, we return the set~$\Xplus \cup X'$ that yields the minimum cost $\alpha(\Xplus) - \beta(\Xplus, X')$.
  A pseudocode is given in~\autoref{alg:inftyFPTr}. 
\end{proof}

\subsection{Special Case: Influence Maximization}\label{section:everythingActive}
In this section, we consider the special case of \probEffectors, called \textsc{Influence Maximization}, where all nodes are targets ($A=V$).
Note that in this case the variant with unlimited budget and the
parameterization by the number of target nodes are irrelevant.

In the influence maximization case, on deterministic instances, one should intuitively choose 
effectors among the ``sources'' of the influence graph, that is, nodes 
without incoming arcs (or among strongly connected components without incoming arcs). 
Moreover, the budget~$b$ bounds the number of sources that can be selected,
and the cost~$c$ bounds the number of sources that can be left out. 
In the following theorem,
we prove that deterministic \probEffectors remains intractable even if either one of
these parameters is small,
but, by contrast,
having $b+c$ as a parameter yields fixed-parameter tractability in the deterministic case.
We mention that the first statement is proven by a reduction from the $\wtwo$-hard \textsc{Set Cover} problem,
while the second statement is proven by a reduction from the $\wone$-hard \textsc{Independent Set} problem.

\begin{theorem}
  \label{thm:infmax}
  \mbox{}
  \begin{enumerate}
    \item \textsc{Influence Maximization},
      parameterized by the maximum number $b$ of effectors,
      is $\wtwo$-hard, even if $G$ is a deterministic ($r=0$) DAG.
    \item \textsc{Influence Maximization},
      parameterized by the cost $c$,
      is $\wone$-hard, even if $G$ is a deterministic ($r=0$) DAG.
    \item If $r = 0$,
      then \textsc{Influence Maximization} can be solved in~$O(\binom{b+c}{b} \cdot (n + m))$ time.
  \end{enumerate}
\end{theorem}

\begin{proof}
  We begin with the first statement,
  namely,
  that \textsc{Influence Maximization}
  (which is equivalent to \probEffectors where all nodes are target nodes, that is, where $A = V$),
  parameterized by the maximum number $b$ of effectors,
  is $\wtwo$-hard, even if $G$ is a deterministic ($r=0$) DAG.
  We provide a parameterized reduction from the $\wtwo$-complete \probSetCover problem~\cite{DF13}.
  \probDef
    {\probSetCover}
    {Sets $S = \{S_1,\ldots,S_m\}$ over elements $U = \{u_1,\ldots,u_n\}$, and parameter $h \in \N$.}
    {Is there $S' \subseteq S$ such that $|S'| = h$ and $\bigcup_{S_i\in S'} S_i = U$?}
  Given an instance for \probSetCover, we
  create an instance for \textsc{Influence Maximization} as follows.
  Add a node $v_{S_j}$ for each set $S_j$ and write $V_S=\{v_{S_i}\mid S_i\in S\}$.
  Add a node $v_{u_i}$ for each element $u_i$ and write $V_U=\{v_{u_i}\mid u_i\in U\}$.
  For each $u_i \in S_j$, add an arc $v_{S_j} \to v_{u_i}$ with influence
  probability~$1$.
  Set $b:=h$, $c:=m - h$, and $A:= V_S\cup V_U$.
  
  We can assume that any solution~$X$ is such that  $X\subseteq V_S$ and  $|X|=b$. 
  Note that all nodes of $V_U$ are activated if and only if $S':=\{S_i\mid v_{S_i}\in V_{S'}\}$ is a set cover for $U$.
  Hence, any solution with cost 
  $c=|V_S\setminus X|=m -h$ needs to pay only for the unselected nodes of $V_S$, and yields a set cover of $U$.
  Reversely, for any set cover~$S'$ for~$U$ of size~$h$, the set~$X:=\{v_{S_i}\mid S_i\in S'\}$
  is a set of effectors with cost at most~$c$.
  
  We continue with the second statement,
  namely,
  that \textsc{Influence Maximization},
  parameterized by the cost $c$,
  is $\wone$-hard, even if $G$ is a deterministic ($r=0$) DAG.
  We provide a reduction from the following $\wone$-complete problem~\cite{DF13}.
  \probDef
    {\probMaxIndependentSet}
    {An undirected graph  $G = (V,E)$ and parameter $k \in \N$.}
    {Is there an independent set $I \subseteq V$ (i.e., no edge has both endpoints in~$I$) such that $|I|\geq  k$?}
  Consider an instance $(G=(V,E), k)$ of \probMaxIndependentSet. Create an influence graph as follows:
  For each vertex $v\in V$, add a node $n_v$ and for each edge~$e\in E$, add a node $n_e$. 
  Let $N_V:=\{n_v\mid v\in V\}$ and $N_E:=\{n_e\mid e\in E\}$.
  Add an arc $n_v \to n_e$ with influence probability~1 for each edge $e$ incident to vertex $v$ in $G$. 
  Set $c:=k$, $b:=|V|-k$, and $A:=N_V\cup N_E$.
 
  Consider any solution with cost $c$. 
  Note that we can assume any optimal solution to choose only nodes from $N_V$, since for any edge~$e=\{u,v\}\in E$ it is always better to choose either~$n_v$ or~$n_u$ instead of the node~$n_e$.
  Write $X\subseteq N_V$ for the set of effectors, $N_I=N_V\setminus X$, and $I=\{v\mid n_v\in N_I\}$.
  We have $|X|\le b$ and $|I|=|N_I|=|V|-|X|\ge |V|-b= k=c$.
  Since the cost equals~$c$, it follows that~$|N_I|=c$ and only the nodes in~$N_I$ are left inactive. Hence, no edge $e$ has both endpoints in $I$. That is, $I$ is an independent set of size $c=k$. 
  Conversely, any independent set $I$ of $G$ directly translates into a set of effectors $X=\{n_v\mid v\in V\setminus I\}$ for the created influence graph.  

  We finish with the third statement, namely, that if $r = 0$,
  then \textsc{Influence Maximization} can be solved in~$O(\binom{b+c}{b} \cdot (n + m))$ time.
  To start with, let~$G'$ be the condensation of $G$ (that is, the DAG obtained by contracting each strongly connected component (SCC) of~$G$ into one node). 
  Note that since $r=0$, we can assume that any minimal solution selects at most one node from each SCC in~$G$. Moreover, it does not matter which node of an SCC is selected since they all lead to the same activations.
  Hence, in the following, we solve \textsc{Influence Maximization} on the condensation~$G'$,
  where selecting a node means to select an arbitrary node in the corresponding SCC in~$G$.
  
  Let $R$ denote the set of nodes of $G'$ with in-degree zero. 
  Note that any node in~$R$ not chosen as an effector yields a cost of at least~1 in~$G$, since the nodes in its corresponding SCC cannot be activated by in-neighbors.
  Hence, we can assume that $|R|\leq b+c$, because otherwise the instance is a ``no''-instance.
  Moreover, we can assume that all effectors are chosen from~$R$. 
  Indeed, consider any solution selecting a node~$u\notin R$ as effector. 
  Then, $u$ has at least one in-neighbor~$v$ and selecting~$v$ instead yields the same number of effectors, while the cost can only be reduced (since at least as many nodes are activated). 
  Since the graph~$G'$ is a DAG, repeating this process yields a solution with smaller cost having only effectors in~$R$.
  Hence, it is sufficient to enumerate all possible choices of
  size-$b$ subsets of the $b+c$ nodes in $R$, 
  and check in polynomial time whether the chosen set of effectors in~$G$
  yields a cost of at most~$c$.  
\end{proof}

% \paragraph{Treewidth as a Parameter.}
% As \probEffectors is in general not polynomial-time solvable (unless~$\p=\np$ (\autoref{thm:inftyNP})),
% but polynomial-time solvable on trees,
% it is natural to consider the treewidth of the underlying undirected graph as a parameter.
% Indeed, treewidth is a well-known concept in algorithmic graph theory. Informally, treewidth measures
% how ``tree-like'' a graph is---trees have treewidth one.
% We note that for deterministic influence graphs ($r=0$) under the influence maximization model ($A=V$), 
% \probEffectors is equivalent to a special case of a related problem, namely {\sc Target Set Selection} (with constant thresholds),
% for which fixed-parameter tractability for the parameter treewidth is
% already known~\cite{Ben-ZwiHLN11, CNNW14}. 
% We are confident that by a very tedious but straightforward extension
% of the proof of this fixed-parameter tractability result one can also
% obtain fixed-parameter tractability with respect to treewidth for the case
% where some nodes are non-targets ($A\subsetneq V$). We refrain from
% presenting any details.
% Moreover, we conjecture that, for influence graphs with $r>0$ probabilistic arcs, the problem is still fixed-parameter 
% tractable for the combined parameter treewidth and~$r$.  The most
% challenging open question is whether \probEffectors is
% fixed-parameter tractable when parameterized by the treewidth, even
% with an unbounded number of probabilistic arcs.

\subsection{Results in Contradiction with \citet{LTGMH10}.}\label{sect:flawDetails}
The following two claims of \citet{LTGMH10} are contradicted by the results presented in our work.

According to~\citet[Lemma 1]{LTGMH10}, in the {\sc Influence Maximization} case with cost value $c=0$, 
\probEffectors is NP-complete.
The reduction is incorrect: it uses a target node $\ell$ which influences all other vertices with probability 1 (in at most two steps). 
It suffices to select $\ell$ as an effector in order to activate all
vertices, so such instances always have a trivial solution
($X=\{\ell\}$), and the reduction collapses.
On the contrary, we prove in our \autoref{lem:det_c0_p} that all instances with $c=0$ can be solved in linear time. 

According to the discussion of~\citet{LTGMH10} following
their Corollary~1, there exists a polynomial-time algorithm for \probEffectors with deterministic instances (meaning $r=0$).
Their model coincides with our model in the case of 
{\sc Influence Maximization}. However, the given algorithm is flawed: it does not consider the influence 
\emph{between} different strongly connected components. Indeed, as we prove in \autoref{thm:infmax}, 
finding effectors under the deterministic model is NP-hard, even in
the case of {\sc Influence Maximization}.

\section{Conclusion}
Inspired by work of \citet{LTGMH10}, we contributed a fine-grained computational complexity analysis of a ``non-monotone version'' of finding effectors in networks.
Indeed, we argued why we believe this to be at least as natural as the more restricted ``monotone model'' due to \citet{LTGMH10}. A particular case for this is that we may find more robust solutions, that is, solutions that are more resilient against noise.
The central point is that, other than \citet{LTGMH10}, we allow non-target nodes to be effectors as well, motivated by the assumption that knowledge about the state of a node may get lost from time to time (see \autoref{sec:Model} for further discussion). Altogether, we observed that both models suffer from computational hardness even in very restricted settings. For the case of unlimited budget, we believe that both models coincide with respect to solvability and hence with respect to a fine-grained computational complexity classification.

Our work is of purely theoretical and classification nature. One message for practical solution approaches we can provide is that it may help to get rid of some probabilistic arcs by rounding them up to 1 (making them deterministic) or rounding them down to~0 (deleting the arcs)---this could be interpreted as some form of approximate computation of effectors. Network structure restrictions seem to be less promising since our hardness results even hold for directed acyclic graphs. Still, there is hope for finding further islands of tractability, for instance by ignoring budget constraints and restricting the degree of randomness.

We leave several challenges for future research.
First, it remains to prove or disprove \autoref{con:infinite-budget}.
Moreover,
while we considered most of the parameterizations for most of the variants of the \probEffectors problem,
we have left some work for future research,
specifically the parameterized complexity of \probEffectors where there is infinite budget and arbitrarily-many probabilistic arcs
(see the corresponding open question in~\autoref{table:results}).
%Also, parameterizing \probEffectors by the treewidth of the
%underlying undirected graph deserves further investigation. 
A further,
more general direction would be to consider other diffusion models and other cost functions.
For example,
it is also natural to maximize the probability that precisely the current activation state is achieved
when selecting the effectors to be initially active.
Moreover,
it seems as if the current diffusion model and its somewhat complicated probabilistic nature
is one of the main reasons for the intractability of our problem.
It would be interesting to consider other diffusion models,
possibly simpler ones,
and see whether it is possible to push the tractability results to apply for more cases.
Specifically,
it would be interesting to extend our results concerning the parameter ``degree of randomness'' to such models.

%
%\begin{acknowledgements}
%We are grateful to two anonymous reviewers of \emph{Theory of Computing Systems} whose careful and constructive feedback helped %to significantly improve the presentation of the paper.
%\end{acknowledgements}

\bibliographystyle{abbrvnat}
\bibliography{bibliography}

\end{document}